\documentclass[a4paper,UKenglish,thm-restate,autoref,cleveref,numberwithinsect,pdfa]{lipics-v2021}

\pdfoutput=1
\hideLIPIcs
\nolinenumbers

\author{Jakob Piribauer}{Technische Universit\"at Dresden, Germany; Universit\"at Leipzig, Germany}{jakob.piribauer@tu-dresden.de}{0000-0003-4829-0476}{}

\authorrunning{J. Piribauer}



\title{Demonic variance and a non-determinism score for Markov decision processes}

\usepackage[T1]{fontenc}
\usepackage{graphicx}

\usepackage{amsmath,amssymb}
\usepackage{tikz}
\usetikzlibrary{positioning,automata,fit,shapes,calc}
\usepackage{dsfont}
\usepackage{thmtools,thm-restate}
\usepackage{pgfplots}
\usepgfplotslibrary{fillbetween}
\pgfplotsset{width=10cm,compat=1.9}
\usepackage{tabularx}
    \newcolumntype{L}{>{\raggedright\arraybackslash}X}
    \newcolumntype{R}{>{\raggedleft\arraybackslash}X}
\usepackage{wrapfig}

\usepackage{thm-restate}
\usepackage{todonotes}
\usepackage{subcaption}

\newtheorem{mydef}{Definition}[section]

\newtheorem{assumption}[mydef]{Assumption}

\usepackage{amssymb}
\usepackage{tikz}

\newcommand{\rawdiaplus}{%
  \begin{tikzpicture}
    \useasboundingbox (-0.7ex, -0.9ex) rectangle (0.7ex, 0.9ex);
    \node (w) at (-0.7ex,0) {};
    \node (e) at (+0.7ex,0) {};
    \node (s) at (0,-0.9ex) {};
    \node (n) at (0,+0.9ex) {};
    \draw (n.center) -- (e.center) -- (s.center) -- (w.center) -- (n.center);
    \draw (n.center) -- (s.center);
    \draw (e.center) -- (w.center);
  \end{tikzpicture}}

\newsavebox{\diamondplusbox}
\savebox{\diamondplusbox}{\rawdiaplus}

\newcommand{\rawdiaminus}{%
  \begin{tikzpicture}
    \useasboundingbox (-0.7ex, -0.9ex) rectangle (0.7ex, 0.9ex);
    \node (w) at (-0.7ex,0) {};
    \node (e) at (+0.7ex,0) {};
    \node (s) at (0,-0.9ex) {};
    \node (n) at (0,+0.9ex) {};
    \draw (n.center) -- (e.center) -- (s.center) -- (w.center) -- (n.center);
    \draw (e.center) -- (w.center);
  \end{tikzpicture}}

\newsavebox{\diamondminusbox}
\savebox{\diamondminusbox}{\rawdiaminus}

\newcommand{\WR}{\mathrm{WR}}

\newcommand{\Var}{\mathbb{V}}

\newcommand{\dem}{\mathit{dem}}
\newcommand{\nds}{\mathrm{NDS}}

\newcommand{\cE}{\mathcal{E}}

\newcommand{\cM}{\mathcal{M}}
\newcommand{\cN}{\mathcal{N}}

\newcommand{\eqdef}{\ensuremath{\stackrel{\text{\tiny def}}{=}}}

\renewcommand{\Pr}{\mathrm{Pr}}

\newcommand{\sinit}{s_{\mathit{\scriptscriptstyle init}}}

\newcommand{\Act}{\mathit{Act}}

\newcommand{\act}{\alpha}

\newcommand{\fpath}{\pi}

\newcommand{\last}{\mathit{last}}

\newcommand{\Cyl}{\mathit{Cyl}}

\newcommand{\sched}{\mathfrak{S}}
\newcommand{\tsched}{\mathfrak{T}}

\newcommand{\usched}{\mathfrak{U}}

\newcommand{\rsched}{\mathfrak{R}}

\newcommand{\rew}{\mathit{rew}}
\newcommand{\wgt}{\mathit{wgt}}

\newcommand{\Rational}{\mathbb{Q}}

\newcommand{\CiteAppendix}[1]{}

\begin{document}

\maketitle

\begin{abstract}
This paper studies the influence of probabilism and non-determinism on some quantitative aspect $X$ of the execution of a system modeled as a Markov decision process (MDP).
To this end,  the novel notion of \emph{demonic variance} is introduced:
For a random variable $X$ in an MDP $\cM$, it is defined as $1/2$ times the maximal expected squared distance of the values of $X$ in two independent execution of $\cM$ in which also the non-deterministic choices are resolved independently by two distinct schedulers.

It is shown that the demonic variance is between $1$ and $2$ times as large as the maximal variance of $X$ in $\cM$ that can be achieved by a single scheduler.
This allows defining a non-determinism score for $\cM$ and $X$ measuring how strongly the difference of $X$ in two  executions of $\cM$ can be influenced by the non-deterministic choices.
Properties of MDPs $\cM$  with extremal values of the non-determinism score are established.
Further, the algorithmic problems of computing the maximal variance and the demonic variance are investigated for two random variables, namely weighted reachability and accumulated rewards. In the process, also the structure of schedulers maximizing the variance and of scheduler pairs  realizing the demonic variance is analyzed.
\end{abstract}

\section{Introduction}

	  \begin{figure*}[t]
	  \Large 
	  \centering
		  \begin{subfigure}[b]{0.36\textwidth}
\centering
   \resizebox{0.85\textwidth}{!}{%
      \begin{tikzpicture}[scale=1,auto,node distance=8mm,>=latex]
        \tikzstyle{round}=[thick,draw=black,circle]

 \node[round, draw=black,minimum size=9mm] (choice2) {};
  \node[round, left=12mm of choice2, minimum size=9mm] (choice1) {};

    \node[round, right=12mm of choice2, minimum size=9mm] (choice3) {};
   \node[round, above=6mm of choice2, minimum size=9mm,xshift=-80pt] (outcome1) {$1$};
     \node[left=6mm of outcome1] (label1) {\Large $\cM$:};

   \node[round, above=6mm of choice2, minimum size=9mm,xshift=-30pt] (outcome0) {$0$};
    \node[round, above=6mm of choice2, minimum size=9mm,xshift=80pt] (outcome3) {$3$};
   \node[round, above=6mm of choice2, minimum size=9mm,xshift=30pt] (outcome4) {$4$};
    \node[round, below=6mm of choice2, minimum size=9mm] (init) {$\sinit$};

         \draw[color=black ,->,very thick] (init)  edge    node [pos=0.3,left=3pt] {$\alpha$}  (choice1) ;
           \draw[color=black ,->,very thick] (init)  edge    node [pos=0.3,right] {$\beta$}  (choice2) ;
                      \draw[color=black ,->,very thick] (init)  edge    node [pos=0.3,right=3pt] {$\gamma$}  (choice3) ;
           
            \draw[color=black ,->,very thick] (choice2)  edge  node [very near start, anchor=center] (m6) {}   node [pos=0.1,left=3pt] {$1/2$}  (outcome0) ;
             \draw[color=black ,->,very thick] (choice2)  edge  node [very near start, anchor=center] (m5)  {} node [pos=0.1,right=3pt] {$1/2$}  (outcome4) ;
             \draw[color=black , very thick] (m5.center) edge [bend right=35]  (m6.center);

\draw[color=black ,->,very thick] (choice1)  edge     (outcome1) ;

\draw[color=black ,->,very thick] (choice3)  edge    (outcome3) ;

      \end{tikzpicture}
    }
  \end{subfigure}
    \begin{subfigure}[b]{0.36\textwidth}
\centering
    \resizebox{0.85\textwidth}{!}{%
      \begin{tikzpicture}[scale=1,auto,node distance=8mm,>=latex]
        \tikzstyle{round}=[thick,draw=black,circle]

 \node[round, draw=black,minimum size=9mm] (choice2) {};

  \node[round, left=12mm of choice2, minimum size=9mm] (choice1) {};

    \node[round, right=12mm of choice2, minimum size=9mm] (choice3) {};
   \node[round, above=6mm of choice2, minimum size=9mm,xshift=-80pt] (outcome1) {$4$};
      \node[left=6mm of outcome1] (label1) {\Large $\cN$:};
   \node[round, above=6mm of choice2, minimum size=9mm,xshift=-30pt] (outcome0) {$0$};
    \node[round, above=6mm of choice2, minimum size=9mm,xshift=80pt] (outcome3) {$0$};
   \node[round, above=6mm of choice2, minimum size=9mm,xshift=30pt] (outcome4) {$4$};
    \node[round, below=6mm of choice2, minimum size=9mm] (init) {$\sinit$};

         \draw[color=black ,->,very thick] (init)  edge    node [pos=0.3,left=3pt] {$\alpha$}  (choice1) ;
           \draw[color=black ,->,very thick] (init)  edge    node [pos=0.3,right] {$\beta$}  (choice2) ;
                      \draw[color=black ,->,very thick] (init)  edge    node [pos=0.3,right=3pt] {$\gamma$}  (choice3) ;
           
            \draw[color=black ,->,very thick] (choice2)  edge  node [very near start, anchor=center] (m6) {}   node [pos=0.1,left=3pt] {\LARGE$\frac12$}  (outcome0) ;
             \draw[color=black ,->,very thick] (choice2)  edge  node [very near start, anchor=center] (m5)  {} node [pos=0.1,right=3pt] {\LARGE$\frac12$}  (outcome4) ;
             \draw[color=black , very thick] (m5.center) edge [bend right=35]  (m6.center);

    \draw[color=black ,->,very thick] (choice1)  edge  node [very near start, anchor=center] (k6) {}   node [pos=0.1,left=3pt] {\LARGE$\frac14$}  (outcome1) ;
             \draw[color=black ,->,very thick] (choice1)  edge  node [very near start, anchor=center] (k5)  {} node [pos=0.1,right=3pt] {\LARGE$\frac34$}  (outcome0) ;
             \draw[color=black , very thick] (k5.center) edge [bend right=35]  (k6.center);
             
                 \draw[color=black ,->,very thick] (choice3)  edge  node [very near start, anchor=center] (h6) {}   node [pos=0.1,left=3pt] {\LARGE$\frac34$}  (outcome4) ;
             \draw[color=black ,->,very thick] (choice3)  edge  node [very near start, anchor=center] (h5)  {} node [pos=0.1,right=3pt] {\LARGE$\frac14$}  (outcome3) ;
             \draw[color=black , very thick] (h5.center) edge [bend right=35]  (h6.center);

      \end{tikzpicture}
    }
  \end{subfigure}
\caption{ MDPs modeling a communication protocol and the time required to process a message.}
\label{fig:communication}
\end{figure*}
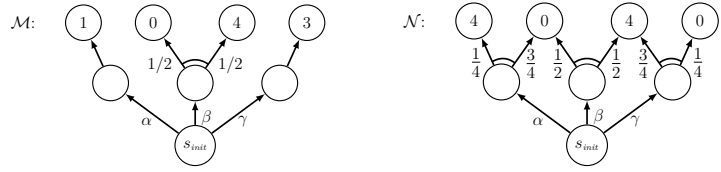

In software and hardware systems, uncertainty manifests in two distinct forms: \emph{non-determinism} and \emph{probabilism}. Non-determinism  emerges from, e.g., unknown operating environments, user interactions, or concurrent processes. Probabilistic behavior  arises through deliberate randomization in algorithms or can be inferred, e.g., from  probabilities of component failures.
 In this paper, we investigate the uncertainty in the value $X$ of some quantitative aspect of a system whose behavior is subject to non-determinism \emph{and} probabilism.
 On the one hand, we aim to quantify this uncertainty. In the spirit of the variance that quantifies uncertainty in purely probabilistic settings, we introduce the notion of \emph{demonic variance} that generalizes the variance in the presence of non-determinism.
 On the other hand, we provide a \emph{non-determinism score} (NDS) based on this demonic variance that measures the extent to which the uncertainty of $X$ can be ascribed to the non-determinism.

 As formal models, we use Markov decision processes (MDPs, see, e.g., \cite{Puterman}), one of the most prominent models combining non-determinism and probabilism, heavily used in verification, operations research, and artificial intelligence.
 The non-deterministic choices in 
an MDP are resolved by a \emph{scheduler}.
 Once a scheduler is fixed, the system behaves purely probabilistically.

\vspace{6pt}
\noindent\textbf{Demonic variance.}
For a random variable $Y$, the variance is equal to half the expected squared deviation of two independent copies $Y_1$ and $Y_2$ of $Y$: 
\[
\Var(Y) \eqdef \mathbb{E}((Y-\mathbb{E}(Y))^2)= \mathbb{E}(Y^2)-\mathbb{E}(Y)^2 = \frac{1}{2}\mathbb{E}(Y_1^2 - 2 Y_1Y_2 +Y_2^2) = \frac{1}{2}\mathbb{E}((Y_1-Y_2)^2) .
\]
For a quantity $X$ in an MDP $\cM$, we obtain a random variable $X_{\cM}^\sched$ for each scheduler $\sched$.\footnote{Note that the notation $X^{\sched}_{\cM}$ here differs from the notation used in the technical part of the paper.}
 The 
maximal variance
$
\Var^{\max}_{\cM} (X) \eqdef \sup_{\sched} \Var (X^{\sched}_{\cM}) 
$
can serve as a measure for the ``amount of probabilistic uncertainty'' regarding $X$ present in the MDP. 
However, in the presence of non-determinism, quantifying the spread of outcomes in terms of the squared deviation of two independent executions of a system 
gives rise to a whole new meaning:
We can allow the non-determinism to be resolved independently as well.
To this end, we consider two different scheduler $\sched_1$ and $\sched_2$ in two independent copies $\cM_1$ and $\cM_2$ of $\cM$ and define 
\[
\Var^{\sched_1,\sched_2}_{\cM}(X) \eqdef \frac{1}{2}\mathbb{E}((X^{\sched_1}_{\cM_1}-X^{\sched_2}_{\cM_2})^2).\]
If we now allow for a \emph{demonic} choice of the two schedulers making this uncertainty as large as possible, we arrive at the \emph{demonic~variance} 
$\Var^{\dem}_{\cM}(X) \eqdef \sup_{\sched_1,\sched_2} \Var^{\sched_1,\sched_2}_{\cM}(X)$
of $X$ in $\cM$.

\begin{example}
\label{ex:communication}
To illustrate a potential use case,
consider a
communication network in which messages are processed according to a randomized protocol employed on different hardware at the different nodes of the network.
A low worst-case expected processing time of the protocol is clearly desirable. In addition, however, large differences  in the processing time at different nodes 
make buffering necessary and  increase the risk of package losses.

Consider  the  MDPs $\cM$ and $\cN$ in Fig. \ref{fig:communication} modeling such a communication protocol.
Initially, a non-deterministic choice between $\alpha$, $\beta$, and $\gamma$ is made. Then, a final node containing the processing time $X$ is reached according to the depicted distributions.
In both MDPs, the expected value of $X$ lies between $1$ and $3$ for all schedulers $\sched$ -- with the values $1$ and $3$ being realized by $\alpha$ and $\gamma$. Furthermore, as the outcomes lie between $0$ and $4$, the distribution over outcomes leading to the highest possible variance of $4$ is the one that takes value $0$ and $4$ with probability $\frac12$ each, which is realized by a scheduler choosing $\beta$. So, $\Var^{\max}_{\cM}(X)=\Var^{\max}_{\cN}(X)=4$.

However, the demonic variances are different: Our results will show that the demonic variance is obtained by a pair of \emph{deterministic} schedulers that do not randomize over the non-deterministic choices.
In $\cM$, we can easily check that no combination of such schedulers $\sched$ and $\tsched$ leads to a value $\Var^{\sched,\tsched}_{\cM}(X)$ of more than $4=\Var^{\beta,\beta}_{\cM}(X)$ where $\beta$ denotes the scheduler that chooses $\beta$ with probability $1$.
In $\cN$, on the other hand, the demonic variance is
$
\Var^{\dem}_{\cN}(X)=\Var^{\alpha,\gamma}_{\cN}(X) = \frac12 \mathbb{E}((X_{\cN_1}^\alpha - X_{\cN_2}^\gamma)^2) =\frac12 \left(\frac{10}{16}\cdot 16\right) = 5$.

So, despite the  same  maximal variance and range of  expected values,
the worst-case expected squared deviation between two values of $X$ in independent executions is worse in $\cN$ than in $\cM$.
Hence, we  argue that the protocol modeled by $\cM$ should be preferred.
\end{example}

\vspace{6pt}
\noindent\textbf{Non-determinism score (NDS).}
By the definition of the demonic variance, it is clear that $\Var^{\dem}_{\cM}(X) \geq \Var^{\max}_{\cM}(X)$.
Under mild assumptions ensuring the well-definedness, we will  prove that  $\Var^{\dem}_{\cM}(X) \leq 2 \Var^{\max}_{\cM}(X)$, too.
So, the demonic variance is between 1 and 2 times as large as the maximal variance. We use this to define the
\emph{non-determinism score (NDS)}
\[
\nds(\cM,X) \eqdef \frac{\Var^{\dem}_{\cM}(X) - \Var^{\max}_{\cM}(X)}{ \Var^{\max}_{\cM}(X)} \in [0,1].
\]
The NDS captures how much larger the expected squared deviation of two outcomes can be made by resolving the non-determinism in two executions independently compared to how large it can be solely due to the probabilism under 
a single resolution of the non-determinism.

		  \begin{figure*}[t]
		  \begin{subfigure}[b]{0.24\textwidth}
\centering
    \resizebox{.85\textwidth}{!}{%
      \begin{tikzpicture}[scale=1,auto,node distance=8mm,>=latex]
        \tikzstyle{round}=[thick,draw=black,circle]

        \node[round, draw=black,minimum size=9mm] (goal) {$2$};
        \node[round, left=12mm of goal, minimum size=9mm] (s2) {$3$};
             \node[round, right=12mm of goal, minimum size=9mm] (s3) {$0$};
        \node[round, below=6mm of goal, xshift=-40pt, minimum size=9mm] (c1) { };
        \node[round, below=6mm of goal, xshift=40pt, minimum size=9mm] (c2) { };        
        \node[round, below=24mm of goal, minimum size=9mm] (init) {$\sinit$};
        
         \draw[color=black ,->,very thick] (init)  edge  node [very near start, anchor=center] (k6) {}   node [pos=0.3,left=3pt] {$1/2$}  (c1) ;
           \draw[color=black ,->,very thick] (init)  edge    node [very near start, anchor=center] (k5)  {}   node [pos=0.3,right=3pt] {$1/2$}  (c2) ;
           \draw[color=black , very thick] (k5.center) edge [bend right=35]  (k6.center);
           
            \draw[color=black ,->,very thick] (c1)  edge  node [very near start, anchor=center] (m6) {}   node [pos=0.1,left=3pt] {$1/3$}  (s2) ;
             \draw[color=black ,->,very thick] (c1)  edge  node [very near start, anchor=center] (m5)  {} node [pos=0.1,right=3pt] {$2/3$}  (goal) ;
             \draw[color=black , very thick] (m5.center) edge [bend right=35]  (m6.center);
             
              \draw[color=black ,->,very thick] (c2)  edge  node [very near start, anchor=center] (n6) {}   node [pos=0.1,left=3pt] {$1/2$}  (goal) ;
             \draw[color=black ,->,very thick] (c2)  edge   node [very near start, anchor=center] (n5)  {}  node [pos=0.1,right=3pt] {$1/2$}  (s3) ;
              \draw[color=black , very thick] (n5.center) edge [bend right=35]  (n6.center);

      \end{tikzpicture}
    }
  \caption{$\nds = 0$.}
  \label{fig:nds0}
  \end{subfigure}
    \begin{subfigure}[b]{0.24\textwidth}
\centering
    \resizebox{.85\textwidth}{!}{%
      \begin{tikzpicture}[scale=1,auto,node distance=8mm,>=latex]
        \tikzstyle{round}=[thick,draw=black,circle]

        \node[round, draw=black,minimum size=9mm] (goal) {$2$};
        \node[round, left=12mm of goal, minimum size=9mm] (s2) {$3$};
             \node[round, right=12mm of goal, minimum size=9mm] (s3) {$0$};
        \node[round, below=6mm of goal, xshift=-40pt, minimum size=9mm] (c1) { };
        \node[round, below=6mm of goal, xshift=40pt, minimum size=9mm] (c2) { };        
        \node[round, below=24mm of goal, minimum size=9mm] (init) {$\sinit$};
        
         \draw[color=black ,->,very thick] (init)  edge    node [pos=0.3,left=3pt] {$\alpha$}  (c1) ;
           \draw[color=black ,->,very thick] (init)  edge    node [pos=0.3,right=3pt] {$\beta$}  (c2) ;
           
            \draw[color=black ,->,very thick] (c1)  edge  node [very near start, anchor=center] (m6) {}   node [pos=0.1,left=3pt] {$1/3$}  (s2) ;
             \draw[color=black ,->,very thick] (c1)  edge  node [very near start, anchor=center] (m5)  {} node [pos=0.1,right=3pt] {$2/3$}  (goal) ;
             \draw[color=black , very thick] (m5.center) edge [bend right=35]  (m6.center);
             
              \draw[color=black ,->,very thick] (c2)  edge  node [very near start, anchor=center] (n6) {}   node [pos=0.1,left=3pt] {$1/2$}  (goal) ;
             \draw[color=black ,->,very thick] (c2)  edge   node [very near start, anchor=center] (n5)  {}  node [pos=0.1,right=3pt] {$1/2$}  (s3) ;
              \draw[color=black , very thick] (n5.center) edge [bend right=35]  (n6.center);

      \end{tikzpicture}
    }
  \caption{$\nds \approx 0.32$.}
  \label{fig:nds3}
  \end{subfigure}
   \begin{subfigure}[b]{0.24\textwidth}
\centering
    \resizebox{.85\textwidth}{!}{%
      \begin{tikzpicture}[scale=1,auto,node distance=8mm,>=latex]
        \tikzstyle{round}=[thick,draw=black,circle]

        \node[round, draw=black,minimum size=9mm] (goal) {$2$};
        \node[round, left=12mm of goal, minimum size=9mm] (s2) {$3$};
             \node[round, right=12mm of goal, minimum size=9mm] (s3) {$0$};
        \node[round, below=6mm of goal, xshift=-40pt, minimum size=9mm] (c1) { };
        \node[round, below=6mm of goal, xshift=40pt, minimum size=9mm] (c2) { };        
        \node[round, below=24mm of goal, minimum size=9mm] (init) {$\sinit$};
        
         \draw[color=black ,->,very thick] (init)  edge    node [pos=0.3,left=3pt] {$\alpha$}  (c1) ;
           \draw[color=black ,->,very thick] (init)  edge    node [pos=0.3,right=3pt] {$\beta$}  (c2) ;
           
            \draw[color=black ,->,very thick] (c1)  edge  node [very near start, anchor=center] (m6) {}   node [pos=0.1,left=3pt] {$1/3$}  (s2) ;
             \draw[color=black ,->,very thick] (c1)  edge  node [very near start, anchor=center] (m5)  {} node [pos=0.1,right=3pt] {$2/3$}  (goal) ;
             \draw[color=black , very thick] (m5.center) edge [bend right=35]  (m6.center);
             
              \draw[color=black ,->,very thick] (c2)  edge    node [pos=0.1,right=3pt] {$1$}  (s3) ;

      \end{tikzpicture}
    }
  \caption{$\nds \approx 0.92$.}
  \label{fig:nds9}
  \end{subfigure}
     \begin{subfigure}[b]{0.24\textwidth}
\centering
    \resizebox{.85\textwidth}{!}{%
      \begin{tikzpicture}[scale=1,auto,node distance=8mm,>=latex]
        \tikzstyle{round}=[thick,draw=black,circle]

        \node[round, draw=white,minimum size=9mm] (goal) {};
        \node[round, left=12mm of goal, minimum size=9mm] (s2) {$3$};
             \node[round, right=12mm of goal, minimum size=9mm] (s3) {$0$};
        \node[round, below=6mm of goal, xshift=-40pt, minimum size=9mm] (c1) { };
        \node[round, below=6mm of goal, xshift=40pt, minimum size=9mm] (c2) { };        
        \node[round, below=24mm of goal, minimum size=9mm] (init) {$\sinit$};
        
         \draw[color=black ,->,very thick] (init)  edge    node [pos=0.3,left=3pt] {$\alpha$}  (c1) ;
           \draw[color=black ,->,very thick] (init)  edge    node [pos=0.3,right=3pt] {$\beta$}  (c2) ;
           
            \draw[color=black ,->,very thick] (c1)  edge    node [pos=0.1,left=3pt] {$1$}  (s2) ;
                         
              \draw[color=black ,->,very thick] (c2)  edge    node [pos=0.1,right=3pt] {$1$}  (s3) ;

      \end{tikzpicture}
    }
  \caption{$\nds =1$.}
  \label{fig:nds1}
  \end{subfigure}
\caption{Example MDPs with different non-determinism scores (NDSs).}
\label{fig:nds}
\end{figure*}
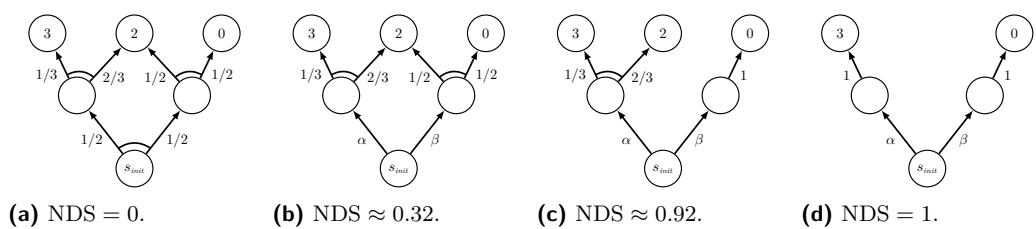

\begin{example}
For an illustration of the NDS, four simple MDPs and their NDSs are depicted in Figure \ref{fig:nds}. In all of the  MDPs except for the first one, a scheduler has to make a (randomized) choice over actions $\alpha$ and $\beta$ in the initial state $\sinit$.
Afterwards one of the terminal states is reached according to the specified probabilities. The terminal states are equipped with a weight that specifies the value of $X$ at the end of the execution.
For all of these MDPs, the maximal variance can be computed by expressing the variance in terms of the probability $p$ that $\alpha$ is chosen and maximizing the resulting quadratic function. In the interest of brevity, we do not present these computations. The pair of (deterministic) schedulers realizing the demonic variance 
always consists of the scheduler choosing $\alpha$ and the scheduler choosing $\beta$ making it  easy to compute the demonic variance in these examples.
\end{example}

\vspace{6pt}
\noindent\textbf{Potential applications.}
First of all, the demonic variance might serve as the basis for refined guarantees on the behavior of systems, in particular, when employed in different environments.
As a first  result in this direction, we will prove an analogue to Chebyshev's Inequality using the demonic variance. 
Further, as illustrated in Example \ref{ex:communication}, 
achieving a low demonic variance or NDS can be desirable when designing systems.
Hence, a reasonable synthesis task could be to design a system ensuring a high expected value of a quantity $X$ while keeping the demonic variance of $X$ below a threshold.

Secondly, the demonic variance and the NDS can serve to enhance the 
 explainability of a system's behavior, a topic of growing importance  in the area of formal verification (see, e.g.,  \cite{DBLP:conf/icalp/BaierD0JMPZ21} for an overview).
More concretely, the NDS can be understood as a measure assigning \emph{responsibility} for the scattering of a quantity $X$ in different executions to the non-determinism and the probabilism present in the system, respectively.
Further, considering the NDS for different starting states  makes it possible to pinpoint regions of the state space in which the non-determinism has a particularly high influence.
Notions of {responsibility} that quantify to which extent certain facets of the behavior of a system can be ascribed to certain components, states, or events
have been studied in various settings
 \cite{ChocklerH04,YazdanpanahDJAL19, BaierFM2021, MascleBFJK2021, BaierFM2021b}.

Finally, the NDS can also be understood as a measure for the power of control when non-determinism  models controllable aspects of a system. This interpretation could be useful, e.g., when designing exploration strategies in reinforcement learning.
Here, the task is to learn good strategies as fast as possible  by interacting with a system. One of the main challenges is to decide which regions of the state space to explore   (see \cite{Ladosz22} for a recent survey).
Estimations for the NDS  starting from different states could be useful here:
States from which the  NDS is high might be more promising to explore than
states from which the NDS is low as the difference in received rewards from such a state is largely subject to randomness.

\vspace{6pt}
\noindent\textbf{Contributions.}
Besides establishing general results for the demonic variance and the NDS, we investigate the two notions for weighted reachability and accumulated rewards.
For weighted reachability, terminal states of an MDP are equipped with a weight that is received if an execution ends in this state. For accumulated rewards, all states are assigned  rewards that are summed up along an execution.
The main contributions of this paper are as follows.
\begin{itemize}
\item
We introduce the novel notions of demonic variance and non-determinism score.
For general random variables $X$, we prove that the demonic variance is at most twice as large as the maximal variance.
Furthermore, we prove an analogue of 
Chebyshev's inequality. 
For the non-determinism score, we establish consequences of a score of $0$ or $1$. 
\item
In the process, we prove a result of independent interest using a topology on the space of schedulers that states that convergence with respect to this topology implies convergence of the corresponding probability measures. 
\item
For weighted reachability, we show that the
maximal and the demonic variance  can be computed via quadratic programs.
For the maximal variance, this results in a
polynomial-time algorithm; for the demonic variance, in a separable bilinear program of polynomial size yielding an exponential time upper bound.
Further, we establish that there is a memoryless scheduler maximizing the variance and 
 a pair of memoryless deterministic schedulers realizing the demonic variance.
\item
For accumulated rewards, we prove that 
 the maximal variance and an optimal finite-memory scheduler can be computed in exponential time.
Further, we prove that the demonic variance is realized by a pair of deterministic finite-memory schedulers which can be computed via a bilinear program of exponential size.
\end{itemize}

\vspace{6pt}
\noindent\textbf{Related work.}
We are not aware of investigations of notions similar to the demonic variance for MDPs.
Previous work on the variance in MDPs usually focused on the minimization of the variance.
In \cite{MannorTsitsiklis2011}, the problem to find schedulers that ensure a certain expected value while keeping the variance below a threshold is investigated
for accumulated rewards in the finite horizon setting. It is shown that deciding whether there is a scheduler ensuring variance $0$ is NP-hard.
In \cite{mandl1971variance}, the minimization of the variance of accumulated rewards and of the mean payoff is addressed with a focus on optimality equations and no algorithmic results.
The variance of accumulated weights in Markov chains is shown to be computable in polynomial time in \cite{verhoeff2004reward}.
For the mean payoff, algorithms were given to compute schedulers that achieve given bounds on the expectation and  notions of variance and variability in~\cite{brazdil2017trading}.

One objective incorporating the variance that has been studied on MDPs is the variance-penalized expectation (VPE) \cite{filar1989variance,collins1997finite,PiribauerSB22}. Here, the goal is to find a scheduler that maximizes the expected reward minus a penalty factor times the variance. 
In \cite{PiribauerSB22}, the objective is studied for accumulated rewards.
Methodically, our results for the maximal and demonic variance of accumulated rewards share similarities with the techniques of \cite{PiribauerSB22} and we 
 make use of some results proved there, such as the result that among expectation-optimal schedulers a variance-optimal memoryless deterministic scheduler can be computed in polynomial time.
 Nevertheless, the optimization of the  VPE inherently  requires the minimization of the variance. In particular, it is shown in \cite{PiribauerSB22} that deterministic schedulers are optimal for the VPE, while randomization is necessary for the maximization of the variance.

Besides the variance, several other notions that aim to bound the uncertainty of the outcome of some quantitative aspect in MDPs have been studied -- in particular, in the context of risk-averse optimization:
Given a probability $p$, quantiles for a quantity $X$ are the best bound $B$ such that $X$ exceeds $B$ with probability at  most $p$ in the worst or best case. For accumulated rewards in MDPs, quantiles have been studied
in  \cite{UB13,DBLP:conf/nfm/BaierDDKK14,HaaseKiefer15,DBLP:journals/fmsd/RandourRS17}.
The \emph{conditional value-at-risk} is a more involved measures  that quantifies how far the probability mass of  the tail of the probability distribution lies above a quantile.
In \cite{kretinsky2018}, this notion has been investigated for weighted reachability and mean payoff; in \cite{icalp2020} for accumulated rewards.
A further measure incentivizing a high expected value while keeping the probability of low outcomes small is the {entropic risk} measure. For accumulated rewards, this measure has been studied in \cite{mfcs2023} in stochastic games that extend MDPs with an adversarial  player.

Finally, as the demonic variance is a measure that looks at a system across different executions, there is a conceptual similarity to hyperproperties
\cite{clarkson2010hyperproperties,clarkson2014temporal}. For probabilistic systems, logics expressing hyperproperties that allow to quantify over different executions or schedulers have been introduced in \cite{abraham2018hyperpctl,dimitrova2020probabilistic}.

\section{Preliminaries}
\label{sec:prelim}

\noindent\textbf{Notations for Markov decision processes.}
A \emph{Markov decision process} (MDP) is a tuple $\mathcal{M} = (S,\Act,P,\sinit)$
where $S$ is a finite set of states,
$\Act$ a finite set of actions,
$P \colon S \times \Act \times S \to [0,1] \cap \Rational$  the
transition probability function, and
$\sinit \in S$ the initial state.
We require that
$\sum_{t\in S}P(s,\act,t) \in \{0,1\}$
for all $(s,\alpha)\in S\times \Act$.
We say that action $\alpha$ is \emph{enabled} in state $s$ iff $\sum_{t\in S}P(s,\act,t) =1$ and denote the set of all actions that are enabled in state $s$ by $\Act(s)$. We further require that $\Act(s) \not=\emptyset$ for all $s\in S$.
If for a state $s$ and all actions $\alpha\in Act(s)$, we have $P(s,\alpha,s)=1$, we say that $s$ is \emph{absorbing}.
The paths of $\cM$ are finite or
infinite sequences $s_0 \, \act_0 \, s_1 \, \act_1  \ldots$
where states and actions alternate such that
$P(s_i,\act_i,s_{i+1}) >0$ for all $i\geq0$.
For $\fpath =
    s_0 \, \act_0 \, s_1 \, \act_1 \,  \ldots \act_{k-1} \, s_k$,
$P(\fpath) =
    P(s_0,\act_0,s_1)
    \cdot \ldots \cdot P(s_{k-1},\act_{k-1},s_k)$
denotes the probability of $\fpath$ and $\last(\fpath)=s_k$ its last state. 
Often, we equip MDPs with a reward function $\rew\colon S\times \Act \to \mathbb{N}$.
The \emph{size} of $\cM$
is the sum of the number of states
plus the total sum of the encoding lengths in binary of the non-zero
probability values
$P(s,\alpha,s')$ as fractions of co-prime integers as well as the encoding length in binary of the rewards if a reward function is used.
{A \emph{Markov chain} is an MDP in which the set of actions is a singleton. In this case, we can drop the set of actions and consider a Markov chain as a tuple $\cM=(S,P,\sinit, \rew)$ where 
$P$ now is a function from $S\times S$ to $[0,1]$ and $\rew$ a function from $S$ to $\mathbb{N}$.}

{
An \emph{end component} of $\cM$ is a strongly connected sub-MDP formalized by a subset $S^\prime\subseteq S$ of states and a non-empty subset $\mathfrak{A}(s)\subseteq \Act(s)$  for each state $s\in S^\prime$ such that for each $s\in S^\prime$, $t\in S$ and $\alpha\in \mathfrak{A}(s)$ with $P(s,\alpha,t)>0$, we have $t\in S^\prime$ and such that in the resulting sub-MDP all states are reachable from each other.
An end-component is a $0$-end-component if
it only contains state-action-pairs with reward $0$.}
Given two MDPs $\cM=(S,\Act, P,\sinit)$ and $\cN=(S^\prime,\Act^\prime, P^\prime,\sinit^\prime)$, we define the (synchronous) product $\cM\otimes \cN$ as the tuple 
$(S\times S^\prime, \Act\times \Act^\prime, P^{\otimes}, (\sinit,\sinit^\prime))$ where  we define
$
P^{\otimes} ((s,s^\prime), (\alpha,\beta) , (t,t^\prime)) = P(s,\alpha, t) \cdot P(s^\prime, \beta, t^\prime)
$
for all $(s,s^\prime), (t,t^\prime) \in S\times S^\prime$ and $(\alpha,\beta)\in \Act \times \Act^\prime$.


\noindent\textbf{Schedulers.}
A \emph{ scheduler} (also called \emph{policy}) for $\cM$
is a function $\sched$ that assigns to each finite path $\fpath$
a probability distribution over $\Act(\last(\fpath))$.
If  $\sched(\fpath)=\sched(\fpath^\prime)$ for all finite paths $\fpath$ and $\fpath^\prime$ with $\last(\fpath)=\last(\fpath^\prime)$,
we say that $\sched$ is \emph{memoryless}. In this case, we also view  schedulers as functions mapping states $s\in S$ to probability distributions over $\Act(s)$.
A  scheduler $\sched$ is called deterministic if $\sched(\fpath)$ is a Dirac distribution
for each finite path $\fpath$, in which case we also view the  scheduler as a mapping to actions in $\Act(\last(\fpath))$.
Given two MDPs $\cM=(S,\Act, P,\sinit)$ and $\cN=(S^\prime,\Act^\prime, P^\prime,\sinit^\prime)$ and two schedulers $\sched$ and $\tsched$ for $\cM$ and $\cN$, respectively,
we define the product scheduler $\sched\otimes \tsched$ for $\cM\otimes \cN$ by defining for a finite path $\pi = (s_0,t_0)\, (\alpha_0 ,\beta_0) \, (s_1,t_1)\, \dots \, (s_k,t_k)$:
$
\sched\otimes \tsched (\pi) (\alpha, \beta) = \sched(s_0 \, \alpha_0 \, \dots \, s_k) (\alpha) \cdot  \tsched(t_0 \, \beta_0  \, \dots \, t_k) (\beta)
$
for all $(\alpha, \beta)\in \Act\times \Act^\prime$.

\noindent\textbf{Probability measure.}
We write $\Pr^{\sched}_{\cM,s}$ 
to denote the probability measure induced by a  scheduler $\sched$ and a state $s$ of an MDP $\cM$.
It is defined on the $\sigma$-algebra generated by the {cylinder sets} $\Cyl(\pi)$  of all infinite extensions of a finite path  $\pi =
    s_0 \, \act_0 \, s_1 \, \act_1 \,  \ldots \act_{k-1} \, s_k$ starting in state $s$, i.e., $s_0=s$, by assigning  to $\Cyl(\pi)$ the probability that $\pi$ is realized under $\sched$, which is
   $P^{\sched}(\pi) \eqdef \prod_{i=0}^{k-1}\sched(s_0\, \act_0 \dots \, s_i )(\act_i) \cdot P(s_i,\act_0,s_{i+1}) $. 
This can be extended to a unique probability measure on the mentioned $\sigma$-algebra. For details, see \cite{Puterman}.
For a random variable $X$, i.e., a measurable function defined on  infinite paths in $\cM$, we denote the expected value of $X$ under a  scheduler $\sched$ and state $s$ by $\mathbb{E}^{\sched}_{\cM,s}(X)$.
We define
$\mathbb{E}^{\min}_{\cM,s}(X) \eqdef \inf_{\sched} \mathbb{E}^{\sched}_{\cM,s}(X)$
and
$\mathbb{E}^{\max}_{\cM,s}(X) \eqdef \sup_{\sched} \mathbb{E}^{\sched}_{\cM,s}(X)$.
The variance of $X$ under the probability measure determined by $\sched$ and $s$ in $\cM$ is denoted by $\Var^{\sched}_{\cM,s}(X)$ and defined by
$
    \Var^{\sched}_{\cM,s}(X)\eqdef\mathbb{E}^{\sched}_{\cM,s}((X-\mathbb{E}^{\sched}_{\cM,s}(X))^2)=\mathbb{E}^{\sched}_{\cM,s}(X^2) -\mathbb{E}^{\sched}_{\cM,s}(X)^2$.
We define $\Var^{\max}_{\cM,s}(X) \eqdef \sup_{\sched} \Var^{\sched}_{\cM,s}(X)$.
If $s=\sinit$, we sometimes drop the subscript $s$ in $\Pr_{\cM,s}^{\sched}$,  $\mathbb{E}_{\cM,s}^\sched$ and $\Var^{\sched}_{\cM,s}(X)$.
\vspace{4pt}

\noindent\textbf{Mixing  schedulers.}
Intuitively, we often want to use a  scheduler that initially decides  to behave like a  scheduler $\sched$ and then to stick to this  scheduler with probability $p$ and to behave like a  scheduler $\tsched$ with probability $1-p$. As this intuitive description does not match the definition of  schedulers as functions from finite paths\footnote{This description would be admissible if we allowed stochastic memory updates (see, e.g., \cite{brazdil2014markov}).}, we provide a formal definition:
For two  schedulers $\sched$ and $\tsched$ and $p\in [0,1]$, we use $p\sched \oplus (1-p)\tsched$ to denote the following  scheduler.
For a path $\pi =
    s_0 \, \act_0 \, s_1 \, \act_1 \,  \ldots \act_{k-1} \, s_k$, we define for an action $\alpha$ enabled in $s_k$
    \begin{align*}
     (p\sched \oplus (1-p)\tsched) (\pi) (\alpha)   
     \eqdef \, &\frac{p\cdot P^\sched(\pi)\cdot \sched(\pi)(\alpha)}{p\cdot P^\sched(\pi) + (1-p) \cdot P^{\tsched}(\pi)}  + \frac{(1-p)\cdot P^\tsched(\pi)\cdot \tsched(\pi)(\alpha)}{p\cdot P^\sched(\pi) + (1-p) \cdot P^{\tsched}(\pi)} .
    \end{align*}
   This is well-defined for any path that has positive probability under $\sched$ or $\tsched$. 
   The following result is folklore; a proof is included in Appendix \ref{app:prelim}.
    
    \begin{restatable}{proposition}{propmixing}
    Let the  schedulers $\sched$ and  $\tsched$ and the value $p$ be as above. Then, for any path $\pi =
    s_0 \, \act_0 \, s_1 \, \act_1 \,  \ldots \act_{k-1} \, s_k$, we have
    $
    P^{p\sched \oplus (1-p)\tsched (\pi) }(\pi) = p P^\sched(\pi) + (1-p) P^\tsched(\pi)
    $.
    \end{restatable}

We conclude
$
\Pr_{\cM,s}^{p\sched \oplus (1-p)\tsched} (A) = p\Pr_{\cM,s}^{\sched}(A)  + (1-p) \Pr_{\cM,s}^{\tsched}(A)$
for any measurable set of paths $A$.
Hence, we can  think of the  scheduler $p\sched \oplus (1-p)\tsched$ as behaving like $\sched$ with probability $p$ and like $\tsched$ with probability $(1-p)$.
In particular, we can also conclude that for a random variable $X$, we have 
$
\mathbb{E}_{\cM,s}^{p\sched \oplus (1-p)\tsched} (X) = p\mathbb{E}_{\cM,s}^{\sched}(X)  + (1-p) \mathbb{E}_{\cM,s}^{\tsched}(X)
$.
For the variance, we obtain the following as shown in Appendix \ref{app:prelim}.

\begin{restatable}{lemma}{variancemix}
\label{lem:variance_mix}
Given $\cM$, $X$, and two schedulers $\sched_1$ and $\sched_2$, as well as $p\in [0,1]$, let $\tsched = p \sched_1 \oplus (1-p) \sched_2$. Then,
$
\Var^{\tsched}_{\cM}(X) = p \Var_{\cM}^{\sched_1}(X) + (1-p) \Var_{\cM}^{\sched_2}(X) + p (1-p) (\mathbb{E}_{\cM}^{\sched_1}(X) - \mathbb{E}_{\cM}^{\sched_2}(X))^2$.
\end{restatable}

\noindent
\textbf{Topology and convergence of measures.}
Given a family of topological spaces $((S_i,\tau_i))_{i\in I}$, the product topology $\tau$ on $\prod_{i\in I}S_i$ is the coarsest topology such that the projections
$
p_i \colon \prod_{i\in I}S_i  \to S_i , \quad
(s_i)_{i\in I}  \mapsto s_i
$
are continuous for all $i\in I$.
For measures $(\mu_j)_{j\in \mathbb{N}}$ and $\mu$ on a measure space $(\Omega,\Sigma)$ where $\Omega$ is a metrizable topological space and $\Sigma$ the Borel $\sigma$-algebra on $\Omega$, we say
that the sequence $(\mu_j)_{j\in \mathbb{N}}$ \emph{weakly converges} to $\mu$ if for all bounded continuous functions $f\colon \Omega \to \mathbb{R}$, we have 
$
\lim_{j\to \infty} \int f \mathrm{d}\mu_j = \int f \mathrm{d}\mu$.
The set of infinite paths $\Pi_{\cM}$ of an MDP $\cM$ with the topology generated by the cylinder sets is metrizable as we can define the metric $d(\pi,\pi^\prime)=2^{-\ell}$ where $\ell$ is the length of the longest common prefix of $\pi$ and $\pi^\prime$.

\section{Demonic variance and non-determinism score}
\label{sec:demonic}

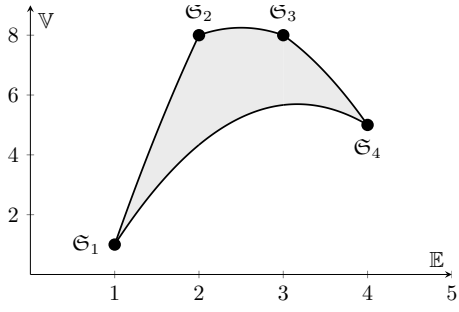
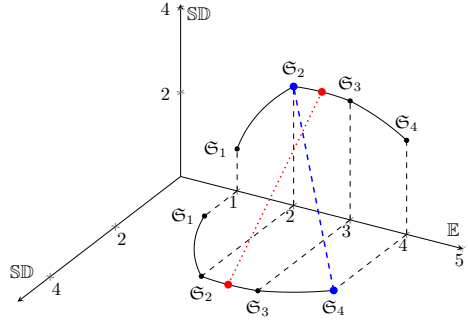
\begin{figure*}[t]
\vspace{-12pt}
\begin{subfigure}[b]{.45\textwidth}
\resizebox{1\textwidth}{!}{
 \begin{tikzpicture}[yscale=0.8,xscale=.8]
      \begin{axis}[
      	height=7cm, width=10cm,
          axis lines = middle,
          xlabel={$\mathbb{E}$},
          ylabel={$\Var$},
          ymin=0, ymax=9,
          xmin=0, xmax=5,
                 ]
        \addplot [name path=A,domain=1:2,
          samples=300,
          color=black,thick]
        {(10*x-x*x-8)};
         \addplot [name path=D,domain=1:2,
          samples=300,
          color=black,thick]
        {-x*x+19/3*x-13/3};
        \addplot [name path=B,domain=2:3,
          samples=300,
          color=black,thick]
        {(6*x-x*x-(x-2))};
         \addplot [name path=E,domain=2:3,
          samples=300,
          color=black,thick]
        {-x*x+19/3*x-13/3};
        \addplot [name path=C,domain=3:4,
          samples=300,
          color=black, thick]
        {(3*x-x*x+9+x-4)};
        \addplot [name path=F,domain=3:4,
          samples=300,
          color=black,thick]
        {-x*x+19/3*x-13/3};

        \addplot[black!8] fill between[of=A and D];
        \addplot[black!8] fill between[of=B and E];
        \addplot[black!8] fill between[of=C and F];

        \node[label={180:{$\sched_1$}},circle,fill,inner sep=2pt] at (axis cs:1,1) {};
        \node[label={90:{$\sched_2$}},circle,fill,inner sep=2pt] at (axis cs:2,8) {};
        \node[label={90:{$\sched_3$}},circle,fill,inner sep=2pt] at (axis cs:3,8) {};
        \node[label={270:{$\sched_4$}},circle,fill,inner sep=2pt] at (axis cs:4,5) {};
      \end{axis}
    \end{tikzpicture}
    }
    \caption{Possible combinations of variance and expectation in Example~\ref{ex:graphical}.}
    \label{fig:2d}
\end{subfigure}
\hspace{12pt}
\begin{subfigure}[b]{.45\textwidth}
\resizebox{1\textwidth}{!}{
\begin{tikzpicture}
\begin{axis}[
     axis lines = center,   
   xmin=0, xmax=5, 
   ymin=0, ymax=5, 
   zmin=0, zmax=4.1,
    xlabel={$\mathbb{SD}$},
    ylabel={$\mathbb{E}$},
    zlabel={$\mathbb{SD}$},
    samples=50,
    view={120}{40} 
]
 \addplot3 [
        domain=1:2,
        samples=60,
        samples y=0,
    ] (
{0},
{x},
{(10*x-x*x-8)^0.5}
); 
\addplot3 [
        domain=2:3,
        samples=60,
        samples y=0,
    ] (
{0},
{x},
{(6*x-x*x-(x-2))^0.5}
); 
\addplot3 [
        domain=3:4,
        samples=60,
        samples y=0,
    ] (
{0},
{x},
{(3*x-x*x+9+x-4)^0.5}
);

 \node[label={180:{$\sched_1$}},circle,fill,inner sep=1pt] at (axis cs:0,1,1) {};
  \draw[dashed] (axis cs:0,1,0) -- (axis cs:0,1,1);
  \node[label={90:{$\sched_2$}},circle,fill,inner sep=1pt] at (axis cs:0,2,8^0.5) {};
    \draw[dashed] (axis cs:0,2,0) -- (axis cs:0,2,8^0.5);
   \node[label={90:{$\sched_3$}},circle,fill,inner sep=1pt] at (axis cs:0,3,8^0.5) {};
       \draw[dashed] (axis cs:0,3,0) -- (axis cs:0,3,8^0.5);
    \node[label={90:{$\sched_4$}},circle,fill,inner sep=1pt] at (axis cs:0,4,5^0.5) {};
        \draw[dashed] (axis cs:0,4,0) -- (axis cs:0,4,5^0.5);

     \addplot3 [
        domain=1:2,
        samples=60,
        samples y=0,
    ] (
    {(10*x-x*x-8)^0.5},
{x},
{0}
); 
\addplot3 [
        domain=2:3,
        samples=60,
        samples y=0,
    ] (
{(6*x-x*x-(x-2))^0.5},
{x},
{0}
); 
\addplot3 [
        domain=3:4,
        samples=60,
        samples y=0,
    ] (
{(3*x-x*x+9+x-4)^0.5},
{x},
{0}
); 

 \node[label={180:{$\sched_1$}},circle,fill,inner sep=1pt] at (axis cs:1,1,0) {};
   \draw[dashed] (axis cs:0,1,0) -- (axis cs:1,1,0);
  \node[label={270:{$\sched_2$}},circle,fill,inner sep=1pt] at (axis cs:8^0.5,2,0) {};
    \draw[dashed] (axis cs:0,2,0) -- (axis cs:8^0.5,2,0);
   \node[label={270:{$\sched_3$}},circle,fill,inner sep=1pt] at (axis cs:8^0.5,3,0) {};
     \draw[dashed] (axis cs:0,3,0) -- (axis cs:8^0.5,3,0);
    \node[label={270:{$\sched_4$}},circle,fill,inner sep=1pt] at (axis cs:5^0.5,4,0) {};
      \draw[dashed] (axis cs:0,4,0) -- (axis cs:5^0.5,4,0);

      \node[color=blue,circle,fill,inner sep=1.5pt] at (axis cs:5^0.5,4,0) {};
         \node[color=blue,circle,fill,inner sep=1.5pt] at (axis cs:0,2,8^0.5) {};
      
         \draw[dashed,thick,color=blue] (axis cs:5^0.5,4,0) -- (axis cs:0,2,8^0.5);
         
         \node[color=red,circle,fill,inner sep=1.5pt] at (axis cs:8.25^0.5,2.5,0) {};
         \node[color=red,circle,fill,inner sep=1.5pt] at (axis cs:0,2.5,8.25^0.5) {};
         
          \draw[dotted,thick,color=red] (axis cs:8.25^0.5,2.5,0) -- (axis cs:0,2.5,8.25^0.5);

\end{axis}
\end{tikzpicture}
}
\caption{Plot of the standard deviation over the expectation on two orthogonal planes.}
\label{fig:3d}
\end{subfigure}
\caption{Graphical illustration of the task to find the demonic variance (see  Example \ref{ex:graphical}).}
\label{fig:graphical}
\end{figure*}

In this section, we formally define the demonic variance.
After proving first auxiliary results, we prove an analogue of Chebyshev's Inequality using the demonic variance.
Then, we introduce the non-determinism score and investigate necessary and sufficient conditions for this score to be $0$ or $1$.
Proofs omitted here can be found in Appendix \ref{app:demonic}.

Throughout this section, let $\cM= (S,\Act, P,\sinit)$ be an MDP and let $X$ be a random variable, i.e., a Borel measurable function on the infinite paths of $\cM$.
We will work under two assumptions that ensure that all notions are well-defined:
First, note that $\Var^{\max}_{\cM}(X)=0$ implies that there is a constant $c$ such that under all schedulers $\sched$, we have $\Pr_{\cM}^{\sched}(X=c)=1$ -- an uninteresting case.
Furthermore, for  meaningful definitions of demonic variance and non-determinism score, we need that the expected value and the variance of $X$ in $\cM$ are finite. Hence, we work under the following assumption:

\begin{assumption}
\label{ass:positivevar}
\label{ass:bounded}
We assume that 
$0<\Var^{\max}_{\cM}(X)<\infty$ and that
$
\sup_{\sched} \left|   \mathbb{E}^{\sched}_{\cM}(X) \right| < \infty$.
\end{assumption}

\subsection{Demonic variance}

As described in the introduction, the idea behind the demonic variance is to quantify the expected squared deviation of $X$ in two independent executions of $\cM$, in which the non-determinism is resolved independently as well.
We use the following notation: Given a path in $\cM\otimes\cM$ consisting of a sequence of pairs of states and pairs of actions, we denote by $X_1$ and $X_2$ the function $X$ applied to the projection of the path on the first component and  on the second component, respectively.
Given two schedulers $\sched_1$ and $\sched_2$ for $\cM$, we define 
\[
\Var^{\sched_1,\sched_2}_{\cM}(X) \eqdef   \frac{1}{2}\mathbb{E}^{\sched_1\otimes \sched_2}_{\cM\otimes \cM} ((X_1 - X_2)^2).
\]
Intuitively, in this definition two independent executions of $\cM$ are run in parallel while the non-determinism  is resolved by $\sched_1$ in the first execution and  by $\sched_2$  in the second component.
As the two components in the  products $\cM\otimes \cM$ and $\sched_1\otimes \sched_2$ are independent, the resulting distributions of $X$ in the two components, i.e., $X_1$ and $X_2$ are independent as well.
The factor $\frac{1}{2}$ is included as for a random variable $Y$, this factor also appears in the representation $\Var(Y)=\frac{1}{2}\mathbb{E}((Y_1-Y_2)^2)$ for two independent copies $Y_1$ and $Y_2$ of $Y$.

The \emph{demonic variance} is now the worst-case value when ranging over all pairs of schedulers:
\[
\Var^{\dem}_{\cM}(X) \eqdef  \sup_{\sched_1,\sched_2} \frac{1}{2}\mathbb{E}^{\sched_1\otimes \sched_2}_{\cM\otimes \cM} ((X_1 - X_2)^2).
\]
A first simple, but useful, result allows us to express $\Var^{\sched_1,\sched_2}_{\cM}(X)$ in terms of the expected values and variances of $X$ under $\sched_1$ and $\sched_2$.

\begin{restatable}{lemma}{variancetwo}
\label{lem:variance_two}
Given two schedulers $\sched_1$ and $\sched_2$ for $\cM$, we have
\[
\Var^{\sched_1,\sched_2}_{\cM}(X)  = \frac{1}{2} \left( \Var_{\cM}^{\sched_1}(X) + \Var_{\cM}^{\sched_2}(X) + (\mathbb{E}_{\cM}^{\sched_1}(X) - \mathbb{E}_{\cM}^{\sched_2}(X))^2 \right).
\]
\end{restatable}

This lemma allows us to provide an insightful graphical interpretation of the demonic variance using 
 the standard deviation $\mathbb{SD}(X)\eqdef \sqrt{\Var(X)}$ of a random variable $X$:
 \begin{example}
 \label{ex:graphical}
 Suppose in an MDP $\cM$, there are four deterministic scheduler $\sched_1,\dots,\sched_4$ with expected values $1$, $2$, $3$, and $4$ and variances $1$, $8$, $8$, and $5$ for a random variable $X$.
 Lemma \ref{lem:variance_mix} allows us to compute the variances of schedulers obtained by randomization leading to parabolic line segments in the expectation-variance-plane as depicted in
 Figure \ref{fig:2d} (see also \cite{PiribauerSB22}). Further randomizations also make it possible to realize any combination of expectation and variance in the interior of the resulting shape.
 When looking for the maximal variance and the demonic variance, only the upper bound of  this shape is relevant.
 
 In Figure \ref{fig:3d},
 we now depict the standard deviations of schedulers on this upper bound over the expectation twice on two orthogonal planes.
 Clearly, the highest standard deviation (and consequently variance) is obtained for the expected value $2.5$ in this example.
 The red dotted line of length $
\sqrt{2\Var^{\max}_{\cM}(X)}$ connects the two points corresponding to this maximum on the two planes.
Considering $\sched_2$ and $\sched_4$, we can also find the value 
 $
\sqrt{2\Var^{\sched_2,\sched_4}_{\cM}(X)}
 $
:
 The blue dashed line connects the point corresponding to $\sched_2$ on one of the planes to the point corresponding to $\sched_4$ on the other plane.
 By the Pythagorean theorem, its length is
 \begin{align*}
 \sqrt{\sqrt{\Var^{\sched_2}_{\cM}(X)}^2 + (\mathbb{E}^{\sched_2}_{\cM}(X) -  (\mathbb{E}^{\sched_4}_{\cM}(X) )^2 + \sqrt{\Var^{\sched_4}_{\cM}(X)}^2  }  = \sqrt{2\Var^{\sched_2,\sched_4}_{\cM}(X)}.
 \end{align*}
 So, finding  $\sqrt{2}$ times the ``demonic standard deviation'' and hence the demonic variance corresponds to finding two points on the two orthogonal graphs  with maximal distance.

 \end{example}

\noindent
The relation between maximal and demonic variance is shown  in the following proposition.
\begin{restatable}{proposition}{boundstwo}
\label{prop:bounds_vdem}
We have
$
\Var^{\max}_{\cM}(X) \leq \Var^{\dem}_{\cM}(X) \leq 2\Var^{\max}_{\cM}(X)$.
\end{restatable}

\noindent
By means of Chebyshev's Inequality, the variance can be used to bound the probability that a random variable $Y$ lies far from its expected value.
Using the demonic variance, we can prove an analogous result providing bounds on the probability that the outcomes of $X$ in two independent executions of the MDP $\cM$ lie far apart.
This can be seen as a first step in the direction of using the demonic variance to provide guarantees on the behavior of a system.
\begin{restatable}{theorem}{cheby}
\label{thm:cheby}

We have
$
\Pr_{\cM\otimes\cM}^{\sched\otimes \tsched} \left( |X_1 - X_2| \geq k\cdot  \sqrt{\Var^{\dem}_{\cM}(X)} \right) \leq \frac{2}{k^2}
$
for any $k\in \mathbb{R}_{>0}$  and  schedulers $\sched$ and $\tsched$ for $\cM$.
\end{restatable}

Using the result that $\Var^{\dem}_{\cM}(X)\leq 2\Var^{\max}_{\cM}(X)$, we obtain the following variant of the inequality providing a weaker bound in terms of the maximal variance.

\begin{restatable}{corollary}{corcheby}
We have $\Pr_{\cM\otimes\cM}^{\sched\otimes \tsched} \left( |X_1 - X_2| \geq k\cdot  \sqrt{\Var^{\max}_{\cM}(X)} \right) \leq \frac{4}{k^2}$
for any $k\in \mathbb{R}_{>0}$  and  schedulers $\sched$ and $\tsched$ for $\cM$.
\end{restatable}

\subsection{Non-determinism score}

We have seen that the demonic variance is larger than the maximal variance by a factor  between $1$ and $2$. As described in the introduction, we use this insight as the basis for a score quantifying how much worse the ``uncertainty'' of $X$ is when non-determinism can be resolved differently in two executions of an MDP compared to how bad it can be in a single execution. We define the non-determinism score (NDS)
\[
\nds(\cM,X) \eqdef \frac{\Var^{\dem}_{\cM}(X) - \Var^{\max}_{\cM}(X)}{\Var^{\max}_{\cM}(X)}.
\]
By Assumption \ref{ass:positivevar},  the NDS is well-defined.
By Proposition \ref{prop:bounds_vdem}, the NDS always returns a value in  $[0,1]$.
Clearly, in Markov chains, the NDS is $0$. A bit more general, we can show:
\begin{restatable}{proposition}{ndszero}
\label{prop:nds0}
If  $\mathbb{E}^{\sched}_{\cM}(X) = \mathbb{E}^{\tsched}_{\cM}(X)$ for all schedulers $\sched$ and $\tsched$, then $\nds(\cM,X)=0$.
\end{restatable}

In transition systems viewed as MDPs in which all transition probabilities are $0$ or $1$, the NDS is $1$: Under Assumption \ref{ass:bounded} in a transition system the value of $X$ must be bounded, i.e., $X\in [a,b]$ for some $a,b\in \mathbb{R}$ such that  $\sup_\pi X(\pi)=b$ and $\inf_\pi X(\pi)=a$ where $\pi$ ranges over all paths. Any path can be realized by a scheduler with probability $1$.
So, for any $\varepsilon>0$, there are schedulers $\sched$ and $\tsched$ with $\Pr^{\sched}_{\cM}(X<a+\varepsilon) = 1$ and $\Pr^{\tsched}_{\cM}(X>b-\varepsilon) = 1$.
Then, $\Var^{\sched,\tsched}_{\cM}(X) \geq\frac{1}{2} (b-a - 2\varepsilon)^2$. For $\varepsilon \to 0$, this converges to $\frac{(a-b)^2}{2}$.
It is well-known that the variance of random variables taking values in $[a,b]$ is maximal for the random variable taking values $a$ and $b$ with probability $\frac{1}{2}$ each.
The variance in this case is $\frac{(a-b)^2}{4}$. So, the maximal variance is (at most) half the demonic variance in this case. Consequently, the NDS is $1$.

Of course, a NDS of $1$ does not imply that there are no probabilistic transitions in $\cM$. 
Nevertheless, a NDS of $1$ has severe implications showing that the outcome of $X$ can be heavily influenced by the non-determinism in this case as the following theorem shows:

\begin{restatable}{theorem}{thmndsone}
\label{thm:nds1}
 If $\nds(\cM,X)=1$, the following statements hold:
 \begin{enumerate}
 \item
 For every $\varepsilon>0$, 
there are schedulers $\mathfrak{Min}_\varepsilon$ and $\mathfrak{Max}_\varepsilon$ with 
$\mathbb{E}^{\mathfrak{Min}_\varepsilon}_{\cM} (X) \leq \mathbb{E}^{\min}_{\cM} (X) +\varepsilon $ and $\Var^{\mathfrak{Min}_\varepsilon}_{\cM} (X) \leq \varepsilon$,  and
$\mathbb{E}^{\mathfrak{Max}_\varepsilon}_{\cM} (X) \geq \mathbb{E}^{\max}_{\cM} (X) -\varepsilon$ and $ \Var^{\mathfrak{Max}_\varepsilon}_{\cM} (X) \leq \varepsilon$.
\item
If there are schedulers $\sched_0$ and $\sched_1$, with $\Var^{\dem}_{\cM}(X) = \Var^{\sched_0,\sched_1}_{\cM}(X)$, then, for $i=0$ or $i=1$, 
$
\Pr_{\cM}^{\sched_i} ( X = \mathbb{E}^{\min}_{\cM}(X)) = 1$ and $ \Pr_{\cM}^{\sched_{1-i}} ( X = \mathbb{E}^{\max}_{\cM}(X)) = 1$.
\item
If $X$ is bounded and continuous wrt the topology generated by the cylinder sets,  there are schedulers $\mathfrak{Min}$ and $\mathfrak{Max}$ with
$
\Pr_{\cM}^{\mathfrak{Min}} ( X = \mathbb{E}^{\min}_{\cM}(X)) = 1$ and  $\Pr_{\cM}^{\mathfrak{Max}} ( X = \mathbb{E}^{\max}_{\cM}(X)) = 1$.
\end{enumerate}
\end{restatable}
\noindent
The first two statements  can be shown by elementary calculations. For  the third statement, we use topological arguments.
We view schedulers as elements of $\prod_{k=0}^\infty \mathrm{Distr}(\Act)^{\mathrm{Paths}_{\cM}^k}$ where $\mathrm{Paths}_{\cM}^k$ is the set of paths of length $k$ in $\cM$
 and prove  the following result:

\begin{proposition}
The space of schedulers 
 $
 \mathrm{Sched}(\cM) = \prod_{k=0}^\infty \mathrm{Distr}(\Act)^{\mathrm{Paths}_{\cM}^k}
 $
with the product topology is compact. So, every sequence of schedulers has a converging subsequence in this space.
Further,
for a sequence $(\sched_j)_{j\in \mathbb{N}}$  converging to a scheduler $\sched$ in this space, the sequence of probability measures
$(\Pr^{\sched_j}_{\cM})_{j\in \mathbb{N}}$ weakly converges to the probability measure $\Pr^{\sched}_{\cM}$.
\end{proposition}

An example for a random variable that is bounded and continuous wrt the topology generated by the cylinder sets is the discounted reward:
Given a reward function $\rew\colon S \to \mathbb{R}$, the discounted reward of a path $\pi=s_0\act_0s_1\dots$ is defined as $\mathit{DR}_\lambda(\pi) \eqdef \sum_{j=0}^\infty \lambda^j \rew(s_j)$
for some discount factor $\lambda\in (0,1)$. First, $| \mathit{DR}_\lambda| $ is bounded by $\max_{s\in S} |\rew(s)| \cdot \frac{1}{1-\lambda}$.
Further, for any $\varepsilon>0$, let $N$ be a natural number such that $\max_{s\in S} |\rew(s)| \cdot \frac{\lambda^N}{1-\lambda}<\varepsilon$. Then, $|\mathit{DR}_\lambda(\pi) - \mathit{DR}_\lambda(\rho) |<\varepsilon$
for all paths $\pi$ and $\rho$ that share a prefix of length more than $N$.

\section{Weighted reachability}
\label{sec:WR}

We now address the problems to compute the demonic and the maximal variance for weighted reachability where a weight is collected on a run depending on which absorbing state is reached.
As the NDS is defined via these two quantities, we do not address it separately here.
Throughout this section, let   $\cM=(S, \Act, P,\sinit)$ be an MDP with  set of absorbing states $T\subseteq S$ and let $\wgt\colon T\to \mathbb{Q}$ be a weight function. 
We define the random variable $\WR$  on infinite paths $\pi$ by
$
\WR(\pi) = 
\wgt(t)$ if $\pi$ reaches  the absorbing state $t\in T$, and 
$\WR(\pi) = 0$ if $\pi$ does not reach $T$.
The main result we are going to establish is the following:
\vspace{6pt}

\noindent \textbf{Main result. }
\textit{
The maximal variance $\Var^{\max}_{\cM}(\WR)$ and an optimal memoryless randomized scheduler can be computed in polynomial time.}

\textit{
The demonic variance $\Var^{\dem}_{\cM}(\WR)$ can be computed as the solution to a bilinear program that can be constructed in polynomial time. Furthermore, there is a pair of memoryless deterministic schedulers realizing the demonic variance.}
\vspace{6pt}

The following standard model transformation collapsing end components (see \cite{deAlf99}) allows us to assume that $T$ is reached almost surely under any scheduler:
We add a new absorbing state $t^\ast$ and set $\wgt(t^\ast)=0$ and collapse all maximal end components $\cE$ in $S\setminus T$ to single states $s_{\cE}$. 
In $s_\cE$, all actions that were enabled in some state in $\cE$ and that did not belong to $\cE$ as well as a new action $\tau$ leading to $t^\ast$ with probability $1$ are enabled. 
In the resulting MDP $\cN$, the set of absorbing states $T\cup\{t^\ast\}$ is reached almost surely under any scheduler. Further, for any scheduler $\sched$ for $\cM$, there is a scheduler $\tsched$ for $\cN$ such that the distribution of $\WR$ is the same under $\sched$ in $\cM$ and under $\tsched$ in $\cN$, and vice versa. So, w.l.o.g., assume the following:
\begin{assumption}
The set $T$ is reached almost surely under any scheduler $\sched$ for $\cM$.
\end{assumption}

In the sequel, we first address the computation of the maximal variance  and afterwards of the demonic variance of $\WR$ in $\cM$.
Omitted proofs can be found in Appendix \ref{app:WR}.

\vspace{6pt}
\noindent
\textbf{Computation of the maximal variance.}
It is well-known that the set of vectors
$
(\Pr^{\sched}_{\cM}(\lozenge q))_{q\in T}
$
of combinations of reachability probabilities for states in $T$ that can be realized by a scheduler $\sched$ can be described by a system of linear inequalities (see, e.g., \cite{Kallenberg}).
We  provide such a system of inequalities below in equations (\ref{eq:nonnegative}) -- (\ref{eq:def}). The equations use 
  variables $x_{s.\alpha}$ for all state-action pairs $(s,\alpha)$ encoding the expected number of times action $\alpha$ is taken in state $s$.
  Setting $\mathds{1}_{s=\sinit}=1$ if $s=\sinit$ and $\mathds{1}_{s=\sinit}=0$ otherwise,     we require
\begin{align}
x_{s,\alpha} & \geq 0  & \text{ for all  $(s,\alpha)$,}
\label{eq:nonnegative} \\
\sum_{\alpha\in \Act(s)} x_{s,\alpha} & = \sum_{t\in S,\beta\in \Act(t)} x_{t,\beta}\cdot P(t,\beta,s) + \mathds{1}_{s=\sinit} 
& \text{ for all $s\in S\setminus T$,}
\label{eq:constraint}\\
y_q & =  \sum_{t\in S,\beta\in \Act(t)} x_{t,\beta}\cdot P(t,\beta,q) & \text{ for all $q\in T$}.
\label{eq:def}
\end{align}
The variables $y_q$ for $q\in T$  represent the probabilities that state $q$ is reached.
We can now express the expected value of $\WR$ and $\WR^2$ via variables $e_1$ and $e_2$ via the constraints:
\begin{equation}
\label{eq:expectations}
e_1=  \sum_{q\in T} y_q \cdot \wgt(q) \quad \text{ and }\quad e_2=  \sum_{q\in T} y_q \cdot \wgt(q)^2.
\end{equation}
The variance can now be written as a quadratic objective function in $e_1$ and $e_2$:
\begin{equation}
\label{eq:objective}
\text{maximize } \quad e_2 - e_1^2.
\end{equation}

\begin{restatable}{theorem}{WRmax}
\label{thm:WRmaxvar}
The  maximal value in objective (\ref{eq:objective}) under constraints (\ref{eq:nonnegative}) -- (\ref{eq:expectations}) is
$\Var^{\max}_{\cM}(\WR)$.
\end{restatable}

Due to the concavity of the objective function, we conclude:

\begin{restatable}{corollary}{corMR}
\label{cor:MR}
The maximal variance $\Var^{\max}_{\cM}(\WR)$ can be computed in polynomial time. Furthermore, there is a memoryless randomized scheduler $\sched$ with  $\Var_{\cM}^{\sched}(\WR)= \Var_{\cM}^{\max}(\WR)$, which can also be computed in polynomial time.
\end{restatable}

\vspace{6pt}
\noindent
\textbf{Computation of the demonic variance.}
The demonic variance can also be expressed as the solution to a quadratic program.
To encode the reachability probabilities for states in $T$ under two distinct schedulers,  we 
use variables $x_{s,\alpha}$ for all state weight pairs $(s,\alpha)$ and $y_q$ for $q\in T$ subject to constraints (\ref{eq:nonnegative}) -- (\ref{eq:def}) as before.
Additionally, we use variables $x_{s,\alpha}^\prime$ for all state weight pairs $(s,\alpha)$ and $y_q^\prime$ for $q\in T$ subject to the analogue constraints (\ref{eq:nonnegative}$^\prime$) -- (\ref{eq:def}$^\prime$) using these primed variables.
The maximization of the demonic variance can be expressed as
\begin{equation}
\label{eq:objdem}
\text{maximize } \quad \frac{1}{2} \sum_{q,r\in T} y_q \cdot y^\prime_r \cdot (\wgt(q)-\wgt(r))^2.
\end{equation}

\begin{restatable}{theorem}{bilinear}
\label{thm:bilinear}
The  maximum  in  (\ref{eq:objdem}) under constraints (\ref{eq:nonnegative}) -- (\ref{eq:def}), (\ref{eq:nonnegative}$^\prime$) -- (\ref{eq:def}$^\prime$)  is
$\Var^{\dem}_{\cM}(\WR)$.
\end{restatable}

The quadratic objective function (\ref{eq:objdem}) is not concave. However, it is \emph{bilinear} and \emph{separable}. 
This means that the variables can be split into two sets,  the primed and the unprimed variables, such that the quadratic terms only contain products of variables from different sets and  each constraint contains 
only variables from the same set.
In general, checking whether the solution to a separable bilinear program exceeds a given threshold is NP-hard \cite{mangasarian1995linear}.
Nevertheless, solution methods tailored for bilinear programs that perform well in practice have been developed (see, e.g., \cite{kolodziej2013global}).
Further, bilinearity allows us to conclude:

\begin{restatable}{corollary}{MDdem}
\label{cor:MD_dem}
There is a pair of memoryless deterministic schedulers $\sched$ and $\tsched$ for $\cM$ such that 
$
\Var_{\cM}^{\dem} (\WR) = \Var^{\sched,\tsched}_{\cM}(\WR)$.
\end{restatable}

For the complexity of the threshold problem, we can conclude an NP upper bound.
Whether the computation of the demonic variance is possible in polynomial time is left open.

\begin{restatable}{corollary}{NPupper}
Given $\cM$, $\wgt$ and $\vartheta\in \mathbb{Q}$, deciding whether $\Var^{\dem}_{\cM}(\WR)\geq \vartheta$ in in NP.
\end{restatable}

\section{Accumulated rewards}
\label{sec:acc}

One of the most important random variables studied on MDPs are accumulated rewards:
Let $\mathcal{M} = (S,\Act,P,\sinit)$ be an MDP and let $\rew\colon S \to \mathbb{N}$ be a reward function.
We extend the reward function to paths $\pi=s_0\act_0s_1\dots$ by
$
\rew(\pi)= \sum_{i=0}^\infty \rew(s_i)$.
For this random variable, we  prove the following result:
\vspace{4pt}

\noindent \textbf{Main result. }
\textit{
The maximal variance $\Var^{\max}_{\cM}(\rew)$ and an optimal  randomized finite-memory scheduler can be computed in exponential time.}

\textit{
The demonic variance $\Var^{\dem}_{\cM}(\rew)$ can be computed as the solution to a bilinear program that can be constructed in exponential time. Furthermore, there is a pair of  deterministic finite-memory schedulers realizing the demonic variance.}
\vspace{4pt}

\noindent
We provide a sketch outlining the proof strategy. For a detailed exposition, see Appendix~\ref{app:acc}.

\begin{proof}[Proof sketch for the main result.]
It can be   checked in polynomial time whether $\mathbb{E}^{\max}_{\cM}(\rew) < \infty$ \cite{deAlf99}.
If this is the case,
this allows us to perform the same preprocessing as in Section \ref{sec:WR} that  removes all end components without changing  the possible distributions of $\rew$ \cite{deAlf99}.

\vspace{6pt}
\noindent
\textbf{Bounding expected values and expectation maximizing actions:}
After the pre-processing, a terminal state is reached almost surely. As shown in \cite{PiribauerSB22}, this allows to obtain a bound $Q$ 
on $\mathbb{E}^{\max}_\cM (\rew^2)$ in polynomial time.
Further, the maximal expectation $\mathbb{E}^{\max}_{\cM,s}(\rew)$ from each state $s$ can be computed in polynomial time \cite{BerTsi91,deAlf99}.
From these values, a set of \emph{maximizing actions} $\Act^{\max}(s)$ for each state $s$ can be computed. After the preprocessing, a scheduler is expectation optimal iff it only chooses actions from these sets.
If a scheduler $\sched$ initially chooses a non-maximizing action in a state $s$, the expected value $\mathbb{E}_{\cM,s}^{\sched}(\rew)$ is strictly smaller than $\mathbb{E}_{\cM,s}^{\max}(\rew)$. We define $\delta$ to be the minimal difference between these values ranging over all starting states and non-maximizing actions.
So, $\delta$ is the ``minimal loss'' in expected value of $\rew$ received by choosing a non-maximizing action.

\vspace{6pt}
\noindent
\textbf{Switching to expectation maximization:}
Using the values $Q$ and $\delta$, we  provide a bound $B$ such that any scheduler  choosing a non-maximizing action with positive probability after a path $\pi$ with $\rew(\pi)\geq B$ cannot realize the maximal variance. 
The bound $B$ can be computed in polynomial time and its numerical value is exponential in the size of the input.

It follows that variance maximizing schedulers have to maximize the future expected rewards after a reward of at least $B$ has been accumulated. Furthermore, we can show that among all expectation maximizing schedulers, a variance maximizing scheduler has to be used above the reward bound $B$.
In \cite{PiribauerSB22}, it is shown that a memoryless deterministic expectation maximizing scheduler $\usched$ that maximizes the variance among all expectation maximizing schedulers can be computed in polynomial time. 
So, schedulers maximizing the variance of $\rew$ can be chosen to behave like $\usched$ once a reward of at least $B$ has been accumulated.

\vspace{6pt}
\noindent
\textbf{Quadratic program:}
Now, we can unfold the MDP $\cM$ by storing in the state space   how much reward has been accumulated  up to the bound $B$. This results on an exponentially larger MDP $\cM^\prime$. Using the expected values $\mathbb{E}^{\usched}_{\cM,s}(\rew)$ and the variances $\Var^{\usched}_{\cM,s}(\rew)$ under $\usched$ from each state $s$, we can formulate a quadratic program similar to the one for weighted reachability in Section \ref{sec:WR} for this unfolded MDP $\cM^\prime$.
From the solution to this  quadratic program, the maximal variance and an optimal memoryless  scheduler $\sched$ for $\cM^\prime$ can be extracted. Transferred back to $\cM$, the scheduler $\sched$ corresponds to a reward-based finite-memory scheduler that keeps track of the accumulated reward up to bound $B$.
As the quadratic program is convex, these computations can be carried out in exponential time.

\vspace{6pt}
\noindent
\textbf{Demonic variance:}
For the demonic variance, the overall proof follows the same steps.
Similar to the bound $B$ above, a bound $B^\prime$ can be provided such that in any pair of scheduler $\sched$ and $\tsched$ realizing the demonic variance, both schedulers can be assumed to switch to the behavior  of the memoryless deterministic scheduler $\usched$ above the reward bound $B^\prime$.
Again by unfolding the state space up to this reward bound, the demonic variance can be computed via a bilinear program of exponential size similar to the one used in 
Section \ref{sec:WR} for weighted reachability. Furthermore, the pair of optimal memoryless deterministic schedulers in the unfolded MDP, which can be extracted from the solution, corresponds to a pair of deterministic reward-based finite-memory schedulers in the original MDP $\cM$.
\end{proof}

\section{Conclusion}
We introduced the notion of demonic variance that quantifies the uncertainty under probabilism \emph{and} non-determinism of a random variable $X$ in an MDP $\cM$.
As this demonic variance is at most twice as big as the maximal variance of $X$, we used it to define the NDS for MDPs.

The demonic variance can be used to provide new types of guarantees on the behavior of  systems. 
 A first step in this direction is the variant of Chebyshev's Inequality using the demonic variance proved in this paper. Furthermore, the demonic variance and the NDS can serve as the basis for notions of responsibility. On the one hand, such notions could ascribe responsibility for the uncertainty to non-determinism and probabilism.
 On the other hand, comparing the NDS from different starting states can be used to identify regions of the state space in which the non-deterministic choices are of high importance. 
 
 For weighted reachability and accumulated rewards, we proved that randomized finite-memory schedulers are sufficient to maximize the variance. For the demonic variance,
 even pairs of deterministic finite-memory schedulers are sufficient.
 While we obtained upper bounds via the formulation of the computation problems as quadratic programs, determining the precise complexities is left as future work.
In the case of accumulated rewards, we restricted to non-negative rewards. When dropping this restriction, severe difficulties have to be expected
as several  related problems on MDPs  exhibit inherent number-theoretic difficulties  rendering the decidability status of the corresponding decision problems open \cite{icalp2020}.

Of course the investigation of the demonic variance and NDS for further random variables constitutes an interesting direction for future work.
For practical purposes, studying also the approximability of the maximal and demonic variance is important.

Finally, In the  spirit of the demonic variance,  further notions can be defined to quantify the uncertainty in $X$ if the non-determinism in two executions of $\cM$ is not resolved independently, but information can be passed between the two executions. This could  be useful, e.g., to analyze the potential power of coordinated attacks on a network.
Formally, such a notion could be defined as
$
\sup_{\sched} \mathbb{E}^{\sched}_{\cM\otimes \cM} ( (X_1-X_2)^2)
$
 where $\sched$ ranges over all schedulers for $\cM\otimes \cM$.
In this context, also using an asynchronous product of $\cM$ with $\cM$ could be reasonable.

%
%
%
 \bibliographystyle{splncs04}
 \bibliography{references/lit}

\begin{thebibliography}{10}
\providecommand{\url}[1]{\texttt{#1}}
\providecommand{\urlprefix}{URL }
\providecommand{\doi}[1]{https://doi.org/#1}

\bibitem{abraham2018hyperpctl}
{\'A}brah{\'a}m, E., Bonakdarpour, B.: Hyperpctl: A temporal logic for
  probabilistic hyperproperties. In: International Conference on Quantitative
  Evaluation of Systems. pp. 20--35. Springer (2018)

\bibitem{deAlf99}
de~Alfaro, L.: Computing minimum and maximum reachability times in
  probabilistic systems. In: 10th International Conference on Concurrency
  Theory (CONCUR). Lecture Notes in Computer Science, vol.~1664, pp. 66--81
  (1999)

\bibitem{mfcs2023}
Baier, C., Chatterjee, K., Meggendorfer, T., Piribauer, J.: Entropic risk for
  turn-based stochastic games. In: Leroux, J., Lombardy, S., Peleg, D. (eds.)
  48th International Symposium on Mathematical Foundations of Computer Science,
  {MFCS} 2023, August 28 to September 1, 2023, Bordeaux, France. LIPIcs,
  vol.~272, pp. 15:1--15:16. Schloss Dagstuhl - Leibniz-Zentrum f{\"{u}}r
  Informatik (2023). \doi{10.4230/LIPICS.MFCS.2023.15},
  \url{https://doi.org/10.4230/LIPIcs.MFCS.2023.15}

\bibitem{DBLP:conf/nfm/BaierDDKK14}
Baier, C., Daum, M., Dubslaff, C., Klein, J., Kl{\"{u}}ppelholz, S.:
  Energy-utility quantiles. In: Badger, J.M., Rozier, K.Y. (eds.) {NASA} Formal
  Methods - 6th International Symposium, {NFM} 2014, Houston, TX, USA, April 29
  - May 1, 2014. Proceedings. Lecture Notes in Computer Science, vol.~8430, pp.
  285--299. Springer (2014). \doi{10.1007/978-3-319-06200-6\_24},
  \url{https://doi.org/10.1007/978-3-319-06200-6\_24}

\bibitem{DBLP:conf/icalp/BaierD0JMPZ21}
Baier, C., Dubslaff, C., Funke, F., Jantsch, S., Majumdar, R., Piribauer, J.,
  Ziemek, R.: From verification to causality-based explications (invited talk).
  In: Bansal, N., Merelli, E., Worrell, J. (eds.) 48th International Colloquium
  on Automata, Languages, and Programming, {ICALP} 2021, July 12-16, 2021,
  Glasgow, Scotland (Virtual Conference). LIPIcs, vol.~198, pp. 1:1--1:20.
  Schloss Dagstuhl - Leibniz-Zentrum f{\"{u}}r Informatik (2021).
  \doi{10.4230/LIPICS.ICALP.2021.1}

\bibitem{BaierFM2021b}
Baier, C., Funke, F., Majumdar, R.: A game-theoretic account of responsibility
  allocation. In: Zhou, Z. (ed.) Proceedings of the Thirtieth International
  Joint Conference on Artificial Intelligence, {IJCAI}. pp. 1773--1779.
  ijcai.org (2021). \doi{10.24963/IJCAI.2021/244}

\bibitem{BaierFM2021}
Baier, C., Funke, F., Majumdar, R.: Responsibility attribution in parameterized
  markovian models. In: Thirty-Fifth {AAAI} Conference on Artificial
  Intelligence, {AAAI} 2021, Thirty-Third Conference on Innovative Applications
  of Artificial Intelligence, {IAAI} 2021, The Eleventh Symposium on
  Educational Advances in Artificial Intelligence, {EAAI} 2021, Virtual Event,
  February 2-9, 2021. pp. 11734--11743. {AAAI} Press (2021).
  \doi{10.1609/AAAI.V35I13.17395},
  \url{https://doi.org/10.1609/aaai.v35i13.17395}

\bibitem{BerTsi91}
Bertsekas, D.P., Tsitsiklis, J.N.: An analysis of stochastic shortest path
  problems. Mathematics of Operations Research  \textbf{16(3)},  580--595
  (1991)

\bibitem{brazdil2014markov}
Br{\'a}zdil, T., Bro{\v{z}}ek, V., Chatterjee, K., Forejt, V., Ku{\v{c}}era,
  A.: Markov decision processes with multiple long-run average objectives.
  Logical Methods in Computer Science  \textbf{10} (2014)

\bibitem{brazdil2017trading}
Brázdil, T., Chatterjee, K., Forejt, V., Kučera, A.: Trading performance for
  stability in {M}arkov decision processes. Journal of Computer and System
  Sciences  \textbf{84},  144--170 (2017).
  \doi{https://doi.org/10.1016/j.jcss.2016.09.009},
  \url{https://www.sciencedirect.com/science/article/pii/S0022000016300897}

\bibitem{ChocklerH04}
Chockler, H., Halpern, J.Y.: {Responsibility and Blame: A Structural-Model
  Approach}. J. Artif. Int. Res.  \textbf{22}(1),  93--115 (Oct 2004)

\bibitem{clarkson2014temporal}
Clarkson, M.R., Finkbeiner, B., Koleini, M., Micinski, K.K., Rabe, M.N.,
  S{\'a}nchez, C.: Temporal logics for hyperproperties. In: Principles of
  Security and Trust: Third International Conference, POST. pp. 265--284.
  Springer (2014)

\bibitem{clarkson2010hyperproperties}
Clarkson, M.R., Schneider, F.B.: Hyperproperties. Journal of Computer Security
  \textbf{18}(6),  1157--1210 (2010)

\bibitem{collins1997finite}
Collins, E.: Finite-horizon variance penalised {M}arkov decision processes.
  Operations-Research-Spektrum  \textbf{19}(1),  35--39 (1997)

\bibitem{dimitrova2020probabilistic}
Dimitrova, R., Finkbeiner, B., Torfah, H.: Probabilistic hyperproperties of
  {M}arkov decision processes. In: International Symposium on Automated
  Technology for Verification and Analysis, ATVA. pp. 484--500. Springer (2020)

\bibitem{feinberg2014convergence}
Feinberg, E.A., Kasyanov, P.O., Zgurovsky, M.Z.: Convergence of probability
  measures and markov decision models with incomplete information. Proceedings
  of the Steklov Institute of Mathematics  \textbf{287}(1),  96--117 (2014)

\bibitem{filar1989variance}
Filar, J.A., Kallenberg, L.C., Lee, H.M.: Variance-penalized {M}arkov decision
  processes. Mathematics of Operations Research  \textbf{14}(1),  147--161
  (1989)

\bibitem{HaaseKiefer15}
Haase, C., Kiefer, S.: The odds of staying on budget. In: 42nd International
  Colloquium on Automata, Languages, and Programming (ICALP). Lecture Notes in
  Computer Science, vol.~9135, pp. 234--246. Springer (2015)

\bibitem{Kallenberg}
Kallenberg, L.: Markov Decision Processes. Lecture Notes. University of Leiden
  (2011)

\bibitem{kolodziej2013global}
Kolodziej, S., Castro, P.M., Grossmann, I.E.: Global optimization of bilinear
  programs with a multiparametric disaggregation technique. Journal of Global
  Optimization  \textbf{57},  1039--1063 (2013)

\bibitem{kozlov1979polynomial}
Kozlov, M.K., Tarasov, S.P., Khachiyan, L.G.: Polynomial solvability of convex
  quadratic programming. In: Doklady Akademii Nauk. vol.~248, pp. 1049--1051.
  Russian Academy of Sciences (1979)

\bibitem{kretinsky2018}
Kret{\'{i}}nsk{\'{y}}, J., Meggendorfer, T.: Conditional value-at-risk for
  reachability and mean payoff in {M}arkov decision processes. In: 33rd Annual
  {ACM/IEEE} Symposium on Logic in Computer Science ({LICS}). pp. 609--618. ACM
  (2018). \doi{10.1145/3209108.3209176},
  \url{http://doi.acm.org/10.1145/3209108.3209176}

\bibitem{Ladosz22}
Ladosz, P., Weng, L., Kim, M., Oh, H.: Exploration in deep reinforcement
  learning: A survey. Inf. Fusion  \textbf{85}(C),  1–22 (sep 2022).
  \doi{10.1016/j.inffus.2022.03.003}

\bibitem{mandl1971variance}
Mandl, P.: On the variance in controlled {M}arkov chains. Kybernetika
  \textbf{7}(1),  1--12 (1971)

\bibitem{mangasarian1995linear}
Mangasarian, O.L.: The linear complementarity problem as a separable bilinear
  program. Journal of Global Optimization  \textbf{6}(2),  153--161 (1995)

\bibitem{MannorTsitsiklis2011}
Mannor, S., Tsitsiklis, J.N.: Mean-variance optimization in {M}arkov decision
  processes. In: Proceedings of the 28th International Conference on Machine
  Learning. p. 177–184. ICML'11, Omnipress, Madison, WI, USA (2011)

\bibitem{MascleBFJK2021}
Mascle, C., Baier, C., Funke, F., Jantsch, S., Kiefer, S.: Responsibility and
  verification: Importance value in temporal logics. In: 36th Annual {ACM/IEEE}
  Symposium on Logic in Computer Science, {LICS}. pp. 1--14. {IEEE} (2021).
  \doi{10.1109/LICS52264.2021.9470597}

\bibitem{icalp2020}
Piribauer, J., Baier, C.: On {S}kolem-hardness and saturation points in
  {M}arkov decision processes. In: Czumaj, A., Dawar, A., Merelli, E. (eds.)
  47th International Colloquium on Automata, Languages, and Programming
  (ICALP). LIPIcs, vol.~168, pp. 138:1--138:17. Schloss
  Dagstuhl--Leibniz-Zentrum f{\"u}r Informatik (2020).
  \doi{10.4230/LIPIcs.ICALP.2020.138},
  \url{https://drops.dagstuhl.de/opus/volltexte/2020/12545}

\bibitem{PiribauerSB22}
Piribauer, J., Sankur, O., Baier, C.: The variance-penalized stochastic
  shortest path problem. In: Bojanczyk, M., Merelli, E., Woodruff, D.P. (eds.)
  49th International Colloquium on Automata, Languages, and Programming,
  {ICALP} 2022, July 4-8, 2022, Paris, France. LIPIcs, vol.~229, pp.
  129:1--129:19. Schloss Dagstuhl - Leibniz-Zentrum f{\"{u}}r Informatik
  (2022). \doi{10.4230/LIPICS.ICALP.2022.129}

\bibitem{Puterman}
Puterman, M.L.: Markov Decision Processes: Discrete Stochastic Dynamic
  Programming. John Wiley \& Sons (1994)

\bibitem{DBLP:journals/fmsd/RandourRS17}
Randour, M., Raskin, J., Sankur, O.: Percentile queries in multi-dimensional
  markov decision processes. Formal Methods Syst. Des.  \textbf{50}(2-3),
  207--248 (2017). \doi{10.1007/S10703-016-0262-7},
  \url{https://doi.org/10.1007/s10703-016-0262-7}

\bibitem{UB13}
Ummels, M., Baier, C.: Computing quantiles in {Markov} reward models. In:
  Pfenning, F. (ed.) 16th International Conference on Foundations of Software
  Science and Computation Structures (FoSSaCS). Lecture Notes in Computer
  Science, vol.~7794, pp. 353--368. Springer (2013)

\bibitem{verhoeff2004reward}
Verhoeff, T.: Reward variance in {M}arkov chains: A calculational approach. In:
  Proceedings of Eindhoven FASTAR Days. Citeseer (2004)

\bibitem{YazdanpanahDJAL19}
Yazdanpanah, V., Dastani, M., Jamroga, W., Alechina, N., Logan, B.: {Strategic
  Responsibility Under Imperfect Information}. In: Proc. of the 18th Intern.
  Conf. on Autonomous Agents and MultiAgent Systems {(AAMAS)}. pp. 592--600.
  AAMAS Foundation (2019)

\end{thebibliography}

\clearpage

\begin{appendix}

\section{Omitted proofs of Section \ref{sec:prelim}}
\label{app:prelim}

\propmixing*

   \begin{proof}
    We prove the result by induction on $k$. For $k=0$, there is nothing to show.
    So, assume the  induction hypothesis (IH) that the claim holds for $k\in \mathbb{N}$.
    Let $\pi =
    s_0 \, \act_0 \, s_1 \, \act_1 \,  \ldots \act_{k-1} \, s_k \act_k \, s_{k+1}$ and $\pi^\prime =
    s_0 \, \act_0 \, s_1 \, \act_1 \,  \ldots \act_{k-1} \, s_k$.
    Then,
    \begin{align*}
    & P^{p\sched \oplus (1-p)\tsched} (\pi) \\
     = & P^{p\sched \oplus (1-p)\tsched} (\pi^\prime) \cdot (p\sched \oplus (1-p)\tsched) (\pi^\prime)(\act_k) \cdot P(s_k,\act_k,s_{k+1}) \\
     \overset{\text{(IH)}}{=} & (p P^\sched(\pi^\prime) + (1-p) P^\tsched(\pi^\prime) ) \cdot P(s_k,\act_k,s_{k+1}) \\
     &  \cdot \left(\frac{p\cdot P^\sched(\pi)\cdot \sched(\pi)(\alpha)}{p\cdot P^\sched(\pi) + (1-p) \cdot P^{\tsched}(\pi)}  + \frac{(1-p)\cdot P^\tsched(\pi)\cdot \tsched(\pi)(\alpha)}{p\cdot P^\sched(\pi) + (1-p) \cdot P^{\tsched}(\pi)}\right)  \\
     = & p P^\sched(\pi)\cdot \sched(\pi)(\alpha)  \cdot P(s_k,\act_k,s_{k+1}) \\
     &+ (1-p) P^\tsched(\pi)\cdot \tsched(\pi)(\alpha) \cdot P(s_k,\act_k,s_{k+1}) \\
     =& pP^\sched(\pi) + (1-p)P^{\tsched}(\pi).
           \end{align*}
This concludes the induction step.  
    \end{proof}

    \variancemix*
    
    \begin{proof}
We express the variance of $X$ under $\tsched$ as 
\begin{align*}
& \Var^{\tsched}_{\cM}(X) = \mathbb{E}^{\tsched}_{\cM}(X^2) -  \mathbb{E}^{\tsched}_{\cM}(X)^2 \\
={}& p\, \mathbb{E}^{\sched_1}_{\cM}(X^2) + (1-p) \mathbb{E}^{\sched_2}_{\cM}(X^2) - \left( p\,\mathbb{E}^{\sched_1}_{\cM}(X) +(1-p) \mathbb{E}^{\sched_2}_{\cM}(X)  \right)^2 \\
={}& p\,\mathbb{E}^{\sched_1}_{\cM}(X^2) - p^2 \mathbb{E}^{\sched_1}_{\cM}(X)^2 -p(1-p)\mathbb{E}^{\sched_1}_{\cM}(X)^2 +p(1-p)\mathbb{E}^{\sched_1}_{\cM}(X)^2 \\
&+ (1-p)\,\mathbb{E}^{\sched_2}_{\cM}(X^2) - (1-p)^2 \mathbb{E}^{\sched_2}_{\cM}(X)^2 \\
& -p(1-p)\mathbb{E}^{\sched_2}_{\cM}(X)^2 +p(1-p)\mathbb{E}^{\sched_2}_{\cM}(X)^2 \\
&- 2 p (1-p) \mathbb{E}^{\sched_1}_{\cM}(X)\mathbb{E}^{\sched_2}_{\cM}(X) \\
={}& p\Var^{\sched_1}_{\cM}(X) + (1-p) \Var^{\sched_2}_{\cM}(X) + p(1-p) (\mathbb{E}^{\sched_1}_{\cM}(X) - \mathbb{E}^{\sched_1}_{\cM}(X))^2.
\end{align*}
\end{proof}

\section{Omitted proofs of Section \ref{sec:demonic}}
\label{app:demonic}

\variancetwo*

\begin{proof}
We compute
\begin{align*}
&  \mathbb{E}^{\sched_1\otimes \sched_2}_{\cM\otimes \cM} ((X_1 - X_2)^2)  \\
 ={} &  \mathbb{E}^{\sched_1\otimes \sched_2}_{\cM\otimes \cM} (X_1^2 -2X_1X_2 + X_2^2)\\
 ={} & \mathbb{E}^{\sched_1}_{\cM} (X^2) - 2  \mathbb{E}^{\sched_1}_{\cM} (X)  \mathbb{E}^{\sched_2}_{\cM} (X) +  \mathbb{E}^{\sched_2}_{\cM} (X^2)  \quad \text{ (independence)} \\
 = {}& \Var^{\sched_1}_{\cM} (X) +  \mathbb{E}^{\sched_1}_{\cM} (X)^2 - 2  \mathbb{E}^{\sched_1}_{\cM} (X)  \mathbb{E}^{\sched_2}_{\cM} (X) +  \Var^{\sched_2}_{\cM} (X) +  \mathbb{E}^{\sched_2}_{\cM} (X)^2  \\
  ={} & \Var^{\sched_1}_{\cM} (X) +  \Var^{\sched_2}_{\cM} (X)  + (\mathbb{E}_{\cM}^{\sched_1}(X) - \mathbb{E}_{\cM}^{\sched_2}(X))^2.  \qedhere
\end{align*}
\end{proof}

\boundstwo*

\begin{proof}
Clearly, $\Var^{\dem}_{\cM}(X) \geq \Var^{\max}_{\cM}(X)$ as for any scheduler $\sched$, we have $\Var^{\sched,\sched}_{\cM}(X) = \Var^{\sched}_{\cM}(X)$.
For a pair of schedulers $\sched$ and $\tsched$, let $\rsched = \frac{1}{2}\sched \oplus \frac{1}{2} \tsched$. Then, by Lemma \ref{lem:variance_mix},
\begin{align*}
& \Var^{\rsched}_{\cM}(X)  = \frac{1}{2} \Var^{\sched}_{\cM}(X) + \frac{1}{2} \Var^{\tsched}_{\cM}(X) + \frac{1}{4} (\mathbb{E}^{\sched}_{\cM}(X) - \mathbb{E}^{\tsched}_{\cM}(X))^2 \\
 \geq{} & \frac{1}{4} \Var^{\sched}_{\cM}(X) + \frac{1}{4} \Var^{\tsched}_{\cM}(X) + \frac{1}{4} (\mathbb{E}^{\sched}_{\cM}(X) - \mathbb{E}^{\tsched}_{\cM}(X))^2  = \frac{1}{2}\Var^{\sched,\tsched}_{\cM}(X)
\end{align*}
where the last equality follows from Lemma \ref{lem:variance_two}. So,
\[
\Var^{\dem}_{\cM}(X) = \sup_{\sched, \tsched}\Var^{\sched,\tsched}_{\cM}(X) \leq 2\sup_{\rsched}\Var^{\rsched}_{\cM}(X) = 2 \Var^{\max}_{\cM}(X). \qedhere
\]
\end{proof}

\cheby*

\begin{proof}
The Markov inequality states that for any random variable $Y$ only taking non-negative values and any value $a>0$, we have $\Pr(Y\geq a) \leq \frac{\mathbb{E}(Y)}{a}$.
Considering the random variable $(X_1-X_2)^2$ in $\cM\otimes\cM$, we can apply the Markov inequality to obtain
\begin{align*}
&\Pr_{\cM\otimes\cM}^{\sched\otimes \tsched} \left( |X_1 - X_2| \geq k\cdot  \sqrt{\Var^{\dem}_{\cM}(X)} \right) \\
& = \Pr_{\cM\otimes\cM}^{\sched\otimes \tsched} \left( (X_1 - X_2)^2 \geq k^2\cdot  \Var^{\dem}_{\cM}(X) \right) \\
&\leq  \frac{\mathbb{E}_{\cM\otimes\cM}^{\sched\otimes \tsched} ((X_1-X_2)^2)}{k^2\cdot  \Var^{\dem}_{\cM}(X) } \leq   \frac{2\Var^{\dem}_{\cM}(X)}{k^2\cdot  \Var^{\dem}_{\cM}(X) } = \frac{2}{k^2}. \qedhere
\end{align*}
\end{proof}

\corcheby*

\begin{proof}
We observe
\begin{align*}
& \Pr_{\cM\otimes\cM}^{\sched\otimes \tsched} \left( |X_1 - X_2| \geq k\cdot  \sqrt{\Var^{\max}_{\cM}(X)} \right) \\
\leq {}& \Pr_{\cM\otimes\cM}^{\sched\otimes \tsched} \left( |X_1 - X_2| \geq k\cdot  \sqrt{\Var^{\dem}_{\cM}(X)/2} \right) \\
={} & \Pr_{\cM\otimes\cM}^{\sched\otimes \tsched} \left( |X_1 - X_2| \geq \frac{k}{\sqrt{2}}\cdot  \sqrt{\Var^{\dem}_{\cM}(X)} \right) \\
\leq {} & \frac{2}{(k/\sqrt{2})^2 }= \frac{4}{k^2}
\end{align*}
where the last line follows from Theorem \ref{thm:cheby}
\end{proof}

\ndszero*

\begin{proof}
Using Lemma \ref{lem:variance_two}, we have for any pair of schedulers $\sched$ and $\tsched$ that
\begin{align*}
\Var^{\sched,\tsched}_{\cM}(X) &= \frac{1}{2} (\Var_{\cM}^{\sched}(X) + \Var_{\cM}^{\tsched}(X) + (\mathbb{E}^{\sched}_{\cM}(X) - \mathbb{E}^{\tsched}_{\cM}(X) )^2) \\
&= \frac{1}{2} (\Var_{\cM}^{\sched}(X) + \Var_{\cM}^{\tsched}(X))
\end{align*}
where the last equality follows from our assumption. So,  the variance $\Var_{\cM}^{\sched}(X)$  or the variance $\Var_{\cM}^{\tsched}(X)$ is at least as large as  $\Var_{\cM}^{\sched,\tsched}(X)$. Consequently,
$\Var^{\max}_{\cM}(X) = \Var^{\dem}_{\cM}(X)$.
\end{proof}

\thmndsone*

\begin{proof}
\phantom{.}

\noindent
(1):
First, note that $\mathbb{E}^{\min}_{\cM}(X)<\mathbb{E}^{\max}_{\cM}(X)$ as otherwise $\nds(\cM,X)=0$ by Proposition \ref{prop:nds0}. Define 
\[
D\eqdef \mathbb{E}^{\max}_{\cM}(X) - \mathbb{E}^{\min}_{\cM}(X).
\]
Let $\varepsilon>0$.
Let $\sched$ and $\tsched$ be two schedulers such that $\Var^{\sched,\tsched}_{\cM}(X) \geq \Var^{\dem}_{\cM}(X) - \delta$ for a $\delta>0$ depending on $D$ and $\varepsilon$, which we will specify later.
Define $\rsched \eqdef \frac{1}{2}\sched \oplus \frac{1}{2} \tsched$. Then, by Lemma \ref{lem:variance_mix},
\begin{align*}
\Var^{\rsched}_{\cM}(X) & = \frac{1}{2} \Var^{\sched}_{\cM}(X) + \frac{1}{2} \Var^{\tsched}_{\cM}(X) + \frac{1}{4} (\mathbb{E}^{\sched}_{\cM}(X) - \mathbb{E}^{\tsched}_{\cM}(X))^2 .
\end{align*}
On the other hand, by Lemma \ref{lem:variance_two},
\begin{align*}
\Var^{\sched,\tsched}_{\cM}(X) &=  \frac{1}{2} \Var^{\sched}_{\cM}(X) + \frac{1}{2} \Var^{\tsched}_{\cM} + \frac{1}{2} (\mathbb{E}^{\sched}_{\cM}(X) - \mathbb{E}^{\tsched}_{\cM}(X))^2.
\end{align*}
As $\nds(\cM,X)=1$, we know $\Var^{\dem}_{\cM}(X) = 2 \Var^{\max}_{\cM}(X) $. So, 
\begin{align*}
&\Var^{\sched,\tsched}_{\cM}(X) \geq  \Var^{\dem}_{\cM}(X) - \delta \geq 2 \Var^{\rsched}_{\cM}(X) - \delta \\
& = \Var^{\sched}_{\cM}(X) +  \Var^{\tsched}_{\cM} (X) + \frac{1}{2} (\mathbb{E}^{\sched}_{\cM}(X) - \mathbb{E}^{\tsched}_{\cM}(X))^2 - \delta \\
& = \Var^{\sched,\tsched}_{\cM}(X) + \frac{1}{2} (\Var^{\sched}_{\cM}(X) +  \Var^{\tsched}_{\cM} (X) ) - \delta.
\end{align*}
So, we conclude
$\Var^{\sched}_{\cM}(X)  \leq 2 \delta$ and $  \Var^{\tsched}_{\cM}(X) \leq 2 \delta$.
Now, let $\mathfrak{Max}$ be a scheduler with $\mathbb{E}^{\mathfrak{Max}}_{\cM}(X) = \mathbb{E}^{\max}_{\cM}(X)$ and 
$\mathfrak{Min}$ a scheduler with $\mathbb{E}^{\mathfrak{Min}}_{\cM}(X) = \mathbb{E}^{\min}_{\cM}(X)$.
Then, 
\begin{align*}
2\Var^{\sched,\tsched}_{\cM}(X) \geq 2\Var^{\mathfrak{Max},\mathfrak{Min}}_{\cM}(X) - 2\delta
\end{align*}
which is equivalent to 
\begin{align*}
&\Var^{\sched}_{\cM}(X) +  \Var^{\tsched}_{\cM} (X) + (\mathbb{E}^{\sched}_{\cM}(X) - \mathbb{E}^{\tsched}_{\cM}(X))^2 \\
\geq &\Var^{\mathfrak{Max}}_{\cM}(X) +  \Var^{\mathfrak{Min}}_{\cM} (X) + (\mathbb{E}^{\mathfrak{Max}}_{\cM}(X) - \mathbb{E}^{\mathfrak{Min}}_{\cM}(X))^2 - 2\delta.
\end{align*}
Plugging in $\Var^{\sched}_{\cM}(X)  \leq 2 \delta$,  $  \Var^{\tsched}_{\cM}(X) \leq 2 \delta$ and $D= \mathbb{E}^{\max}_{\cM}(X) - \mathbb{E}^{\min}_{\cM}(X)$, we obtain
\[
(\mathbb{E}^{\sched}_{\cM}(X) - \mathbb{E}^{\tsched}_{\cM}(X))^2 \geq D^2 - 6\delta.
\]
W.l.o.g., assume $\mathbb{E}^{\sched}_{\cM}(X) \geq \mathbb{E}^{\tsched}_{\cM}(X)$ and let $E\eqdef \mathbb{E}^{\sched}_{\cM}(X) - \mathbb{E}^{\tsched}_{\cM}(X)$.
We conlcude
\[
E\geq \sqrt{D^2 - 6\delta}.
\]
Now, we can specify $\delta$ depending on $\varepsilon$ and $D$. 
We may choose any $\delta>0$ such that $\delta\leq \varepsilon<2$ and $D-\sqrt{D^2 - 6\delta}<\varepsilon$.
Then, $\Var^{\sched}_{\cM}(X)  \leq \varepsilon$ and $  \Var^{\tsched}_{\cM}(X) \leq \varepsilon$.
Further, $E\geq D-\varepsilon$ which implies $\mathbb{E}^{\sched}_{\cM}(X)\geq \mathbb{E}^{\max}_{\cM}(X) - \varepsilon$ and $\mathbb{E}^{\tsched}_{\cM}(X)\leq \mathbb{E}^{\min}_{\cM}(X) + \varepsilon$. 
\vspace{12pt}

\noindent (2):
The same reasoning as in item (1) applied to schedulers $\sched_0$ and $\sched_1$ with $\Var^{\sched_0,\sched_1}_{\cM}(X) = \Var^{\dem}_{\cM}(X)$ allows us first to conclude that 
$\Var^{\sched_0}_{\cM}(X)=0$ and $\Var^{\sched_1}_{\cM}(X)=0$. Then, we obtain 
\[
(\mathbb{E}^{\sched_0}_{\cM}(X) - \mathbb{E}^{\sched_1}_{\cM}(X))^2 = (\mathbb{E}^{\mathfrak{Max}}_{\cM}(X) - \mathbb{E}^{\mathfrak{Min}}_{\cM}(X))^2.
\]
Assuming w.l.o.g., that $\mathbb{E}^{\sched_1}_{\cM}(X)>\mathbb{E}^{\sched_0}_{\cM}(X)$, we get $\mathbb{E}^{\sched_1}_{\cM}(X)=\mathbb{E}^{\max}_{\cM}(X)$ and $\mathbb{E}^{\sched_0}_{\cM}(X)=\mathbb{E}^{\min}_{\cM}(X)$.
Together, this implies 
\[
\Pr^{\sched_0}_{\cM}(X=\mathbb{E}^{\min}_{\cM}(X))=1\quad \text{ and }\quad \Pr^{\sched_1}_{\cM}(X=\mathbb{E}^{\max}_{\cM}(X))=1. 
\] 

\noindent (3):
W.l.o.g., we assume that all actions are enabled in all states. As there is always at least one enabled action, we can simply let all disabled actions in a state have the same transition dynamics as some enabled action.
For a natural number $k$, we let 
\[\Delta_k\eqdef \{f\colon \mathrm{Paths}_{\cM}^k \to \mathrm{Distr}(\Act)\}\]
be the set of functions from the set $\mathrm{Paths}_{\cM}^k$ of finite paths in $\cM$ containing $k$ transitions to the set of distributions $\mathrm{Distr}(\Act)$ over $\Act$.
Viewing $\mathrm{Distr}(\Act)$ as a subset of $[0,1]^{\Act}$, this is a compact space with the usual Euclidean topology as it is a closed subset of $[0,1]^{\Act}$. Similarly, we can view $\Delta_k$ as
 $\mathrm{Distr}(\Act)^{\mathrm{Paths}_{\cM}^k}$ which is a (finite) product of compact spaces and hence compact with the product topology. Note that in this finite product, a basis for the topology is given by 
 products of open sets.
 
 Now, the space of schedulers can be seen as
 \[
 \mathrm{Sched}(\cM) = \prod_{k=0}^\infty \mathrm{Distr}(\Act)^{\mathrm{Paths}_{\cM}^k}.
 \]
 Again, we equip this space with the product topology. By Tychonoff's theorem, this again results in a compact space.

Now, let $(\varepsilon_n)_{n\in \mathbb{N}}$ be a sequence of positive numbers converging to $0$. For each $n$, let $\mathfrak{Min}_{\varepsilon_n}$ be a scheduler as in item (1).
As  $\mathrm{Sched}(\cM)$ is compact, the sequence $(\mathfrak{Min}_{\varepsilon_n})_{n\in \mathbb{N}}$ has a converging subsequence $(\mathfrak{Min}_j)_{j\in \mathbb{N}}$.
Let $\mathfrak{Min}$ be the limit of this subsequence.

Let $\mu_j$ be the probability measure $\Pr^{\mathfrak{Min}_j}_{\cM}$ for each $j\in \mathbb{N}$ and let $\mu$ be the probability measure $\Pr^{\mathfrak{Min}}_{\cM}$.
We will show that $\mu_j$ weakly converges to $\mu$.
As shown in \cite{feinberg2014convergence}, it is sufficient to show for each finite union $U$ of elements of a countable basis of the topology on infinite paths that $\liminf_{j\to\infty} \mu_j(U) \geq \mu(U)$.
The set of cylinder sets forms a countable basis of the topology. So, let $U=\bigcup_{i=1}^\ell \Cyl(\pi_i)$ for finite paths $\pi_1,\dots, \pi_\ell$. W.l.o.g., we can assume that the cylinder sets $\Cyl(\pi_i)$ with $1\leq i \leq \ell$ are disjoint
as we can write $U$ as union of cylinder sets generated by paths of the same length.
So, in fact it is sufficient to prove $\liminf_{j\to\infty} \mu_j(C) \geq \mu(C)$ for all cylinder set $C$.

So, consider a finite path $\pi=s_0\act_0s_1\dots s_k$. Then,
\[
\mu(\Cyl(\pi)) = \Pr^{\mathfrak{Min}}_{\cM}(\pi) =  \prod_{h=0}^{k-1} P(s_i,\act_i,s_{i+1})\cdot \mathfrak{Min}(s_0\act_0s_1\dots s_i)(\act_i).
\]
For each $\delta>0$, the set 
\begin{align*}
S_{\delta} \eqdef \{ & \sched\in \mathrm{Sched}(\cM) \mid  \\
&\left|\mathfrak{Min}(s_0\act_0s_1\dots s_i)(\act_i) - \sched(s_0\act_0s_1\dots s_i)(\act_i)\right|  <\delta \\
&\text{ for all $0\leq i \leq k-1$}\}
\end{align*}
is open in the product topology on $\mathrm{Sched}(\cM)$. So, for each $\delta$, there is an $N$ such that $\mathfrak{Min}_j\in S_{\delta}$ for all $j>N$.
Now, let 
\[
\Delta_\delta \eqdef \sup_{\sched\in S_{\delta}} |\Pr^{\sched}_{\cM}(\pi) - \Pr^{\mathfrak{Min}}_{\cM}(\pi)|.
\]
Then, $\Delta_{\delta}\in \mathcal{O}(\delta)$. This allows us to conclude that for any $\delta^\prime$, there is an $N^\prime$ such that 
$|\Pr^{\mathfrak{Min}_j}_{\cM}(\pi) - \Pr^{\mathfrak{Min}}_{\cM}(\pi)|<\delta^\prime$ for all $j>N^\prime$.
So, $\lim_{j\to \infty} \mu_j(\Cyl(\pi)) = \mu(\Cyl(\pi))$ showing that $\mu_j$ weakly converges to $\mu$. for $j\to \infty$.

For bounded, continuous, Borel measurable functions $X$, we hence can conclude that 
\begin{align*}
 \mathbb{E}^{\mathfrak{Min}}_{\cM} (X) = \int X \mathrm{d} \mu 
= \lim_{j\to \infty} \int X \mathrm{d} \mu_j  =  \lim_{j\to \infty} \mathbb{E}^{\mathfrak{Min}_j}_{\cM} (X)  
 = \mathbb{E}^{\min}_{\cM}(X).
\end{align*}
As also 
$
\int X^2 \mathrm{d} \mu = \lim_{j\to \infty} \int X^2 \mathrm{d} \mu_j
$
and $\lim_{j\to \infty} \Var^{\mathfrak{Min}_j}_{\cM} (X) = 0$,
we  conclude 
$ \mathbb{E}^{\mathfrak{Min}}_{\cM} (X^2) = ( \mathbb{E}^{\min}_{\cM}(X))^2 $.
Hence,
$
\Var^{\mathfrak{Min}}_{\cM} (X) = 0$.
So, 
$
\Pr^{\mathfrak{Min}}_{\cM}(X=  \mathbb{E}^{\min}_{\cM}(X)) =1$.
The existence of the scheduler $\mathfrak{Max}$ as claimed in the theorem can be shown analogously.
\end{proof}

\section{Omitted proofs of Section \ref{sec:WR}}
\label{app:WR}

\WRmax*

\begin{proof}
The correctness of the constraints (\ref{eq:nonnegative}) -- (\ref{eq:def}) is shown, e.g., in \cite[Theorem 9.16]{Kallenberg}. So, for each scheduler $\sched$ there is a solution to (\ref{eq:nonnegative}) -- (\ref{eq:def}) with
\[
y_q = \Pr^{\sched}_{\cM}(\lozenge q)
\]
for all $q\in T$, and vice versa.
This implies directly that  the variables $e_1$ and $e_2$ defined by constraints (\ref{eq:expectations})  in terms of these variables $y_q$ for $q\in T$ satisfy
\[
e_1 = \mathbb{E}^{\sched}_{\cM}(\WR) \quad \text{ and }\quad e_2 = \mathbb{E}^{\sched}_{\cM}(\WR^2)
\] 
for a scheduler $\sched$ corresponding to the values $y_q$.
Hence, any value of the  objective (\ref{eq:objective}) that is obtainable under constraints  (\ref{eq:nonnegative}) -- (\ref{eq:expectations}) is the variance of $\WR$ under some scheduler, and vice versa.
\end{proof}

\corMR*

\begin{proof}
Clearly, the objective function is concave and all constraints are linear.  Hence, the maximal value of the objective function can be computed in polynomial time  \cite{kozlov1979polynomial}.
(Note that  the maximization of a concave function is equivalent to the minimization of a convex function.) 
Furthermore, from the values $x_{s,\alpha}$ in the solution, a memoryless scheduler can be computed by setting $\sched(s)(\alpha) = \frac{x_{s,\alpha}}{\sum_{\alpha\in \Act(s)} x_{s,\alpha}}$ (see, e.g., \cite{Kallenberg}).
\end{proof}

\bilinear*

\begin{proof}
The statement follows analogously to Theorem \ref{thm:WRmaxvar}.
\end{proof}

\MDdem*

\begin{proof}
From Theorem \ref{thm:bilinear}, we can conclude that there are schedulers $\sched$ and $\tsched$ with $\Var^{\dem}_{\cM}(\WR)= \Var^{\sched, \tsched}_{\cM}(\WR)$.
For this fixed scheduler $\tsched$, we can optimize $\Var^{\sched,\tsched}_{\cM}(\WR)$
by fixing the variables $y_q^\prime$ for $q\in T$ in the objective function (\ref{eq:objdem}).
The resulting \emph{linear} program consisting of constraints (\ref{eq:nonnegative}) -- (\ref{eq:def}) and the objective function (\ref{eq:objdem}), is the linear program that computes the 
maximal expected value of the weighted reachability problem with weight function $\wgt^\prime\colon T \to \mathbb{Q}$ given by $\wgt(q) = \sum_{r\in T}  y^\prime_r \cdot (\wgt(q)-\wgt(r))^2$.
As memoryless deterministic schedulers are sufficient to maximize weighted reachability (which is well-known, see, e.g., \cite{Kallenberg}), there is a  memoryless deterministic scheduler $\sched$
maximizing $\Var^{\sched,\tsched}_{\cM}(\WR)$ and hence $\Var^{\dem}_{\cM}(\WR)= \Var^{\sched, \tsched}_{\cM}(\WR)$. 
Analogously, we can show that $\tsched$ can be chosen to be memoryless deterministic.
\end{proof}

\NPupper*

\begin{proof}
We can guess two memoryless deterministic schedulers $\sched$ and $\tsched$. The value $\Var^{\sched,\tsched}_{\cM}(\WR)$ can then easily be computed in polynomial time and be compared to the threshold $\vartheta$.
By Corollary \ref{cor:MD_dem}, this solves the threshold problem in non-deterministic polynomial time.
\end{proof}

\section{Detailed exposition of the results of Section \ref{sec:acc}}
\label{app:acc}

Before we address the computation of the maximal variance and the demonic variance for accumulated rewards, we provide some prerequisites that are well-known.
\vspace{2pt}

\paragraph*{Prerequisites}

We assume that $\mathbb{E}^{\max}_{\cM}(\rew) < \infty$, which can be checked in polynomial time \cite{deAlf99}.
This implies that all reachable end components contain only states with reward $0$. In particular, all reachable absorbing states have reward $0$. Hence, we can perform the same preprocessing as in Section \ref{sec:WR} that introduces a new absorbing state $t^\ast$  and removes all end components. This pre-processing does not change the distributions of $\rew$ that can be realized by a scheduler. For details, see also \cite{deAlf99}.
So, w.l.o.g., we work under the following assumption:
\begin{assumption}
\label{ass:terminal}
We assume that an absorbing state $t$ with $\rew(t)=0$ is reached almost surely under any scheduler $\sched$ for $\cM$.
\end{assumption}

Under our assumption, the maximal expected accumulated reward
$
\mathbb{E}^{\max}_{\cM}(\rew)
$
as well as an optimal memoryless deterministic scheduler can be computed in polynomial time \cite{BerTsi91,deAlf99}.
The actions such a scheduler chooses in a state $s$ belong to the set of \emph{maximizing actions} $\Act^{\max}(s)$ that we define as
\[
\Act^{\max}(s) \\
=  \{\act\in \Act(s) \mid \mathbb{E}^{\max}_{\cM,s}(\rew) = \rew(s,\act) + \sum_{t\in S} P(s,\act,t) \cdot \mathbb{E}^{\max}_{\cM,t}(\rew) \}.
\]
Conversely, under Assumption \ref{ass:terminal} any scheduler only choosing actions in $\Act^{\max}(s)$ in each state $s$ maximize the expected value of $\rew$.
Furthermore, it is known that a bound for the expected value of $\rew^2$ can be computed in polynomial time:
\begin{lemma}[\textnormal{see, e.g.,} \cite{PiribauerSB22}]
\label{lem:boundQ}
Let $p_{\min}$ be the minimal non-zero transition probability in $\cM$ and let $R$ be the largest reward. Then,
\[
\max_{s\in S}\mathbb{E}^{\max}_{\cM,s}(\rew^2) \leq \frac{2\cdot |S|^2\cdot R^2}{p_{\min}{}^{2|S|}} \eqdef Q.
\]
\end{lemma}

Finally, we call a scheduler \emph{reward-based} if for paths $\pi$ and $\pi^\prime$ with $\rew(\pi)=\rew(\pi^\prime)$ and $\last(\pi)=\last(\pi^\prime)$, the scheduler chooses the same distribution over actions.
We say that a scheduler is a \emph{finite-memory} scheduler if it can be implemented with a finite set of memory states that can be updated after each transition according to the state that is reached such that the decisions are based only on the  memory state and the current state of the MDP.

\subsection{Computing the maximal variance}
\label{sec:structvar}
The goal of this section is to prove the following theorem.
\begin{theorem}
The maximal variance $\Var^{\max}_{\cM}(\rew)$ as well as an optimal randomized reward-based finite-memory scheduler can be computed in exponential time.
\end{theorem}

In this section, we will show that we can restrict the supremum in the definitions of the maximal variance \emph{and} the demonic variance to a class of schedulers that switch to a fixed  behavior given by a memoryless deterministic scheduler $\rsched$ once the reward on a path exceeds a computable bound $B$.
Before we prove this result, we introduce additional notation and afterwards define the bound $B$. 
\vspace{4pt}

\paragraph*{Notation}
We  define some notation used below.
Given a scheduler $\sched$ and a finite path $\pi$, we denote by $\sched_\pi$ the residual scheduler of $\sched$ after $\pi$ given by
$
\sched_\pi(\zeta) \eqdef \sched(\pi \circ \zeta)
$
for all finite paths $\zeta$ starting in $\last(\pi)$.
Further, for $\alpha\in \Act(\last(\pi))$, we denote by $\sched_{\pi,\alpha}$ the residual scheduler of $\sched$ after $\pi$ that starts by choosing $\alpha$ and behaves like $\sched$ afterwards. It is given by
assigning probability $1$ to action $\alpha$ initially in state $\last(\pi)$ and for all finite paths $\zeta=\last(\pi)\,\alpha\, \dots$ is given by 
$
\sched_{\pi,\alpha}(\zeta) \eqdef \sched(\pi \circ \zeta)$.
Further, given a scheduler $\tsched$, we denote by $\sched \uparrow_{\pi,\alpha} \tsched$ the scheduler that behaves like $\sched$ unless $\sched$ chooses $\alpha$ at the end of the path $\pi$. 
In this case, $\sched \uparrow_{\pi,\alpha} \tsched$ makes decisions according to $\tsched$ on the paths starting in $\last(\pi)$ from that moment on instead of according to $\sched_{\pi,\alpha}$.
\vspace{2pt}

\paragraph*{Bound $B$}
We begin by defining the  
 bound $B$, after which schedulers can switch to memoryless behavior as we will prove afterwards.
 This bound depends on the maximal expected value $M\eqdef  \max_{s\in S} \mathbb{E}^{\max}_{\cM,s}(\rew)$, the bound $Q$ for the expected value of $\rew^2$ given in Lemma \ref{lem:boundQ} and 
on the minimal loss in expected value $\delta$ that the choice of a non-maximizing action causes, which we define as follows:
\[
\delta \eqdef \min_{s\in S} \min_{\alpha\not \in \Act^{\max}(s)} \mathbb{E}^{\max}_{\cM,s} (\rew) - \sum_{t\in S} P(s,\alpha,t) \mathbb{E}^{\max}_{\cM,t}(\rew).
\]

We define
\[
B \eqdef \frac{Q + \frac{5}{2} M^2}{\delta} +2M +1
\]
The choice of $B$ will become clear later.
By Markov's Inequality, we know that $\Pr^{\sched}_{\cM}(\rew \geq 2 M) \leq \frac{1}{2}$. As $B\geq 2 M$, we conclude $\Pr^{\sched}_{\cM}(\rew\geq B)\leq \frac{1}{2}$.
\vspace{4pt}

\paragraph*{Switch to expectation maximization}
First, we show that the variance of any scheduler $\sched$ that chooses a non-maximizing action after a path with reward more than $B$ can be increased.
Let $\Sigma^{\max}$ be the set of schedulers $\sched$ with $\mathbb{E}^{\sched}_{\cM,s}(\rew) = \mathbb{E}^{\max}_{\cM,s}(\rew) $ for all states $s\in S$.

\begin{theorem}
\label{thm:ermax}
Let $\sched$ be a scheduler for $\cM$ and let $\pi$ be a finite path with $\rew(\pi)\geq B$. Assume that $\pi$ has positive probability under $\sched$ and that  there is an action $\alpha\in \Act(s) \setminus \Act^{\max}(s)$ such that
$\sched(\pi)(\alpha)>0$. 
Let $\tsched\in \Sigma^{\max}$. Then, setting $\rsched \eqdef \sched\uparrow_{\pi,\alpha}\tsched$, we have
\[
\Var^{\sched}_{\cM}(\rew) < \Var^{\rsched}_{\cM}(\rew) .
\]
\end{theorem}

\begin{proof}
Let $A$ be the event that scheduler $\sched$ chooses action $\alpha$ at the end of path $\pi$ and let $p$ be the probability of this event under $\sched$. Note that $\rsched$ switches to the behavior of $\tsched$ exactly in case $A$ occurs. By the definition of $B$, we know that $p\leq\frac{1}{2}$.
Now, we aim to estimate the difference
\begin{align*}
&\Var^{\sched}_{\cM}(\rew) - \Var^{\rsched}_{\cM}(\rew) \\
= & \mathbb{E}^{\sched}_{\cM}(\rew^2) - \mathbb{E}^{\rsched}_{\cM}(\rew^2) 
 + \mathbb{E}^{\rsched}_{\cM}(\rew)^2 -  \mathbb{E}^{\sched}_{\cM}(\rew)^2.
\end{align*}
First, we take a look at the expected values of $\rew^2$:
\begin{align*}
& \mathbb{E}^{\sched}_{\cM}(\rew^2) - \mathbb{E}^{\rsched}_{\cM}(\rew^2) \\
={} & (1-p)(\mathbb{E}^{\sched}_{\cM}(\rew^2\mid \neg A) - \mathbb{E}^{\rsched}_{\cM}(\rew^2\mid \neg A) ) \\
& + p (\mathbb{E}^{\sched}_{\cM}(\rew^2\mid  A) - \mathbb{E}^{\rsched}_{\cM}(\rew^2\mid  A) ) \\
={}&  p (\mathbb{E}^{\sched}_{\cM}(\rew^2\mid  A) - \mathbb{E}^{\rsched}_{\cM}(\rew^2\mid  A) ) 
\end{align*}
as $\sched$ and $\rsched$ behave identically under the condition $\neg A$. Further, letting $W \eqdef \rew(\pi)$,
\begin{align}
&\mathbb{E}^{\sched}_{\cM}(\rew^2\mid  A) - \mathbb{E}^{\rsched}_{\cM}(\rew^2\mid  A) \notag\\
={}& \mathbb{E}^{\sched_{\pi,\alpha}}_{\cM,\last(\pi)}((W+\rew)^2) - \mathbb{E}^{\tsched}_{\cM,\last(\pi)}((W+\rew)^2) \notag\\
={}& W^2 + 2 \cdot  W  \cdot \mathbb{E}^{\sched_{\pi,\alpha}}_{\cM,\last(\pi)}(\rew) + \mathbb{E}^{\sched_{\pi,\alpha}}_{\cM,\last(\pi)}(\rew^2) \notag\\
& - \left( W^2 + 2 \cdot W \cdot \mathbb{E}^{\tsched}_{\cM,\last(\pi)}(\rew) + \mathbb{E}^{\tsched}_{\cM,\last(\pi)}(\rew^2)  \right)  \label{eq:rew2}\tag{$\dagger$}\\
\leq{} & Q - 2\cdot W \cdot  (\mathbb{E}^{\tsched}_{\cM,\last(\pi)}(\rew) -  \mathbb{E}^{\sched_{\pi,\alpha}}_{\cM,\last(\pi)}(\rew)). \notag
\end{align}
Now, we estimate $\mathbb{E}^{\rsched}_{\cM}(\rew)^2 -  \mathbb{E}^{\sched}_{\cM}(\rew)^2$. 
Note that 
\[
\mathbb{E}^{\sched}_{\cM}(\rew) = (1-p)\mathbb{E}^{\sched}_{\cM}(\rew\mid \neg A) + p \mathbb{E}^{\sched}_{\cM}(\rew\mid  A)
\] 
and analogously for $\rsched$.
Let
\[
C \eqdef (1-p)\mathbb{E}^{\sched}_{\cM}(\rew\mid \neg A) \leq \mathbb{E}^{\max}_{\cM}(\rew).
\]
Now, recall that $M= \max_{s\in S} \mathbb{E}^{\max}_{\cM,s}(\rew)$. We compute
\begin{align*}
&\mathbb{E}^{\rsched}_{\cM}(\rew)^2 -  \mathbb{E}^{\sched}_{\cM}(\rew)^2 \\
={} & (C + p \mathbb{E}^{\rsched}_{\cM}(\rew\mid  A))^2 - (C + p \mathbb{E}^{\sched}_{\cM}(\rew\mid  A))^2 \\
={}& (C+p(W+\mathbb{E}^{\tsched}_{\cM,\last(\pi)}(\rew)))^2 - (C+p(W+\mathbb{E}^{\sched_{\pi,\alpha}}_{\cM,\last(\pi)}(\rew)))^2 \\
={}& 2\cdot C \cdot p \cdot ( \mathbb{E}^{\tsched}_{\cM,\last(\pi)}(\rew) - \mathbb{E}^{\sched_{\pi,\alpha}}_{\cM,\last(\pi)}(\rew)) \\
& +p^2 \cdot 2\cdot W \cdot (\mathbb{E}^{\tsched}_{\cM,\last(\pi)}(\rew) - \mathbb{E}^{\sched_{\pi,\alpha}}_{\cM,\last(\pi)}(\rew))  \\
& + p^2 ( \mathbb{E}^{\tsched}_{\cM,\last(\pi)}(\rew)^2 - \mathbb{E}^{\sched_{\pi,\alpha}}_{\cM,\last(\pi)}(\rew)^2) \\
\leq {} & 2   p  M^2 + p^2 \cdot 2 \cdot W \cdot (\mathbb{E}^{\tsched}_{\cM,\last(\pi)}(\rew) -  \mathbb{E}^{\sched_{\pi,\alpha}}_{\cM,\last(\pi)}(\rew)) + p^2  M^2 
\end{align*}
Using that $p<1/2$, we obtain that the last line is at most 
\begin{align*}
p\cdot \left(2   M^2 +  M^2/2 +  W \cdot (\mathbb{E}^{\tsched}_{\cM,\last(\pi)}(\rew) -  \mathbb{E}^{\sched_{\pi,\alpha}}_{\cM,\last(\pi)}(\rew))   \right)
\end{align*}
Put together, we obtain
\begin{align*}
& \Var^{\sched}_{\cM}(\rew) - \Var^{\rsched}_{\cM}(\rew)  \\
\leq {} & p\cdot \big( Q  - 2\cdot W \cdot  (\mathbb{E}^{\tsched}_{\cM,\last(\pi)}(\rew) -  \mathbb{E}^{\sched_{\pi,\alpha}}_{\cM,\last(\pi)}(\rew) ) \\
& + \frac{5}{2} M^2 +  W \cdot (\mathbb{E}^{\tsched}_{\cM,\last(\pi)}(\rew) -  \mathbb{E}^{\sched_{\pi,\alpha}}_{\cM,\last(\pi)}(\rew))   \big) \\
={} &  p\cdot \big(Q + \frac{5}{2} M^2 - W \cdot (\mathbb{E}^{\tsched}_{\cM,\last(\pi)}(\rew) -  \mathbb{E}^{\sched_{\pi,\alpha}}_{\cM,\last(\pi)}(\rew))   \big) \\
\leq {} & p \cdot (Q + \frac{5}{2} M^2 - W\cdot \delta).
\end{align*}
As $W\geq B$, we conclude
$
 \Var^{\sched}_{\cM}(\rew) - \Var^{\rsched}_{\cM}(\rew) <0
$
by the definition of $B$.
\end{proof}

The previous theorem tells us that it is sufficient to consider schedulers that maximize the expected future accumulated reward as soon as a reward of $B$ has been accumulated on a path.
The next question we answer is which of the expectation maximizing schedulers $\tsched \in \Sigma^{\max}$ a variance maximizing scheduler should choose above the bound $B$. The answer can be found in the computations for the proof of the previous theorem. A closer look at Equation (\ref{eq:rew2}) in the previous proof makes it clear that the scheduler $\tsched$ should maximize the expected value of $\rew^2$ among all schedulers in $\Sigma^{\max}$:  Changing the scheduler $\tsched$ among schedulers in $\tsched$ only influences the $\mathbb{E}^{\tsched}_{\cM,\last(\pi)}(\rew^2) $ on the right hand side of Equation (\ref{eq:rew2}).
The following lemma, which is a variation of a result shown in \cite{PiribauerSB22}, shows that such a scheduler can be computed in polynomial time:
\begin{lemma}[\cite{PiribauerSB22}]
\label{lem:rsched}
There is a memoryless deterministic scheduler $\usched \in \Sigma^{\max}$ with 
$
\mathbb{E}^{\usched}_{\cM,s}(\rew^2) = \sup_{\tsched\in \Sigma^{\max}} \mathbb{E}^{\tsched}_{\cM,s}(\rew^2)
$
for all states $s\in S$. Further, $\usched$ as well as $\mathbb{E}^{\usched}_{\cM,s}(\rew^2)$ for all states $s$ can be computed in polynomial time.
\end{lemma}

\begin{proof}
Let $\cM^{\max}$ be the MDP obtained from $\cM$ by only enabling actions in $\Act^{\max}(s)$ in state $s$. It is well-known that any scheduler in $\Sigma^{\max}$ can only schedule actions in $\Act^{\max}$. 
Conversely,
under Assumption \ref{ass:terminal}, any scheduler for $\cM^{\max}$ viewed as a scheduler for $\cM$ is in $\Sigma^{\max}$ (see, e.g., \cite{Kallenberg}).
So, it is sufficient to show that there is a memoryless deterministic scheduler for $\cM^{\max}$ that maximizes the expected value of $\rew^2$ from every state.

In \cite[Lemma 8 and Theorem 9]{PiribauerSB22}, it is shown that in an MDP, in which all schedulers have the same expected value of $\rew$ from any state (as in our MDP $\cM^{\max}$), a memoryless deterministic scheduler minimizing the variance exists and this scheduler as well as its variance can be computed in polynomial time via the definition of a new weight function $\rew^\prime \colon S \to \mathbb{Q}$. While the result is stated for the minimization of the variance in \cite{PiribauerSB22}, 
the same reasoning works for the maximization as well.

Since all schedulers in $\cM^{\max}$ have the same expected value from a state $s$, a scheduler $\sched$ maximizing $\Var^{\sched}_{\cM,s}(\rew)=\mathbb{E}^{\sched}_{\cM,s}(\rew^2) - \mathbb{E}^{\sched}_{\cM,s}(\rew)^2$
also maximizes $\mathbb{E}^{\sched}_{\cM,s}(\rew^2)$. So, the result of \cite{PiribauerSB22} implies that a scheduler $\usched$ as claimed in the lemma can be computed in polynomial time.
\end{proof}

Put together, we can conclude the following theorem:

\begin{theorem}
\label{thm:optimalsched}
The maximal variance can be expressed as 
\[
\Var^{\max}_{\cM}(\rew) = \sup_\sched \Var^{\sched}_{\cM}(\rew)
\]
where 
$\sched$ ranges over all schedulers  such that for all path $\pi$ with $\rew(\pi)\geq B$, we have $\sched_\pi=\usched$ for the memoryless deterministic scheduler $\usched$ given by Lemma \ref{lem:rsched}.
\end{theorem}
Let us denote the set of these schedulers by $\Sigma^{\usched}_{B}$ in the sequel.
As weighted reachability can be seen as a special case of accumulated rewards, we know that randomization is necessary in order to maximize the variance.
While Theorem \ref{thm:optimalsched} does not allow us to reduce the maximization of the variance for accumulated rewards to the case of weighted reachability,
it will nevertheless allow us to construct a quadratic program whose solution is the maximal variance.

To this end, we consider an unfolding of the MDP $\cM$. Let $B$ be the bound provided above and $R$ be the maximal reward occurring in $\cM$ as before.
We define $\cN$ as the MDP that keeps track of the accumulated weight until a weight of more than $B$ has been accumulated.
The state space is given by $S\times \{0,\dots,\lfloor B \rfloor+R\}$. The initial state is $(\sinit,0)$.
The transition probability function $P^\prime$ for $(s,w)\in S\times \{0,\dots, \lfloor B \rfloor \}$ and $\alpha \in \Act$ is given by
		$
		P^{\prime}((s,w),\alpha,(t,v)) = 
		P(s,\alpha,t)$  if $v=w+\rew(s,\alpha)$, and is set to $0$ otherwise.
		All states $(s,w)$ with $w>B$ are made absorbing.
		So, the absorbing states are divided into the two sets 
		\[
		T_1^\prime \eqdef T\times  \{0,\dots, \lfloor B \rfloor \}  \quad \text{ and } T_2^\prime \eqdef S\times \{\lfloor B \rfloor+1, \dots, \lfloor B \rfloor+R\}
		\]
		where $T$ is the set of absorbing states of $\cM$.
		Now, there is a one-to-one correspondence between schedulers for $\cN$ and schedulers in $\Sigma^{\usched}_{B}$:
		In $\cN$ paths are simply equipped with the additional information how much reward has been accumulated, while this information is implicitly given in $\cM$ by the reward function $\rew$.
		 Other than that, the paths with reward $\leq B$ are the same in $\cN$ and $\cM$.
		For paths with length $>B$, schedulers in $\cN$ cannot make choices anymore as they reach an absorbing state. Schedulers in $\Sigma^{\usched}_{B}$ cannot make choices anymore on these paths 
		as they switch to the behavior of $\usched$.
		
		Now, we can derive a system of linear inequalities 
		\begin{equation}
		\label{eq:linear}
		Ax\leq b
		\end{equation}
		for a rational matrix $A$ and vector $b$ computable in polynomial time from $\cN$ such that the variable vector $x$ 
		 contains, among others, variables $y_q$ for each $q\in T_1^\prime \cup T_2^\prime$ such that these variables capture exactly the possible 
		combinations of reachability probabilities of the terminal states in $\cN$. This system of constraints is the system given in Equations  (\ref{eq:nonnegative}) -- (\ref{eq:def}) in  Section \ref{sec:WR} transferred to the MDP $\cN$.
		
		As for weighted reachability, we introduce two new variables $e_1$ and $e_2$ that express the expected value of $\rew$ and of $\rew^2$ and that only depend on the variables $y_q$ with $q\in T_1^\prime \cup T_2^\prime$.
		For the expected value, this is straight-forward:
		\begin{equation}
		\label{eq:exp1}
		e_1 =  \sum_{(s,w)\in T_1} y_{(s,w)} \cdot w +  \sum_{(s,w)\in T_2} y_{(s,w)} \cdot (w+\mathbb{E}^{\max}_{\cM,s}(\rew))
		\end{equation}
		For the expected value of $\rew^2$, we derive
		\begin{align}
		e_2 = {} & \sum_{(s,w)\in T_1} y_{(s,w)} \cdot w^2 \notag \\
		&+  \sum_{(s,w)\in T_2} y_{(s,w)} \cdot (w+ 2w\mathbb{E}^{\max}_{\cM,s}(\rew)+ \mathbb{E}^{\usched}_{\cM,s}(\rew^2))
		\label{eq:exp2}
		\end{align}

		The fact that this set of constraints works as intended is stated in the following lemma and proved in Appendix \ref{app:acc}.
		
\begin{restatable}{lemma}{correctquad}
\label{lem:correctquad}
For each scheduler $\sched \in \Sigma^{\usched}_{B}$, there is a solution to (\ref{eq:linear}) -- (\ref{eq:exp2}) in which
$
e_1 = \mathbb{E}^{\sched}_{\cM}(\rew) $ and $ e_2 = \mathbb{E}^{\sched}_{\cM}(\rew^2)$,
and vice versa.
\end{restatable}

\begin{proof}
Let $A_{s,w}$ be the event that terminal state $(s,w)\in T_1\cup T_2$ is reached in $\cN$. One of the events $A_{s,w}$ occurs with probability $1$ under any scheduler.
In $\cM$, the corresponding events are that either a absorbing state is reached after a path $\pi$ with $\rew(\pi)\leq B$
or that the first prefix of a run with reward $w>B$ ends in state $s$.

Given a scheduler $\sched \in \Sigma^{\usched}_{B}$ viewed as a scheduler for $\cN$, we know (see \cite[Theorem 9.16]{Kallenberg}) that there is a solution to (\ref{eq:linear})
such that $\Pr^{\sched}_{\cM} (\lozenge (s,w)) = y_{(s,w)}$ for all $(s,w)\in T_1\cup T_2$, and vice versa.
In $\cM$, this corresponds to $\Pr^{\sched}_{\cM} (A_{s,w}) = y_{(s,w)}$.

For the expected value, it is then clear that 
\[
\mathbb{E}^{\sched}_{\cM}(\rew)= \sum_{(s,w)\in T_1} y_{(s,w)} \cdot w +  \sum_{(s,w)\in T_2} y_{(s,w)} \cdot (w+\mathbb{E}^{\max}_{\cM,s}(\rew)) = e_1
\]
as the scheduler switches to the behavior of $\usched$ which maximizes the future expected rewards in case of event $A_{s,w}$ for $(s,w)\in T_2$.

For the expected value of $\rew^2$, we compute
\begin{align*}
&\mathbb{E}^{\sched}_{\cM}(\rew)  \\
={} & \sum_{(s,w)\in T_1\cup T_2} \mathbb{E}^{\sched}_{\cM}(\rew^2 \mid A_{s,w}) \cdot  \Pr^{\sched}_{\cM}(A_{s,w}) \\
={} &  \sum_{(s,w)\in T_1} w^2 \cdot y_{s,w} +  \sum_{(s,w)\in T_2} \mathbb{E}^{\sched}_{\cM}(\rew^2 \mid A_{s,w}) \cdot y_{s,w} \\
={}&  \sum_{(s,w)\in T_1} w^2 \cdot y_{s,w}  + \sum_{(s,w)\in T_2} \mathbb{E}^{\usched}_{\cM,s}((w+\rew)^2) \cdot y_{s,w} \\
={} & \sum_{(s,w)\in T_1} w^2 \cdot y_{s,w}  + \sum_{(s,w)\in T_2} \mathbb{E}^{\usched}_{\cM,s}((w+\rew)^2) \cdot y_{s,w} \\
={}& \sum_{(s,w)\in T_1} w^2 \cdot y_{s,w}  \\
& + \sum_{(s,w)\in T_2} (w^2 + 2 w  \mathbb{E}^{\usched}_{\cM,s}(\rew) +  \mathbb{E}^{\usched}_{\cM,s}(\rew^2))  \cdot y_{s,w} \\
={}& e_2. \qedhere
\end{align*}

\end{proof}

Consequently, the maximal variance can be found via the following optimization objective subject to  (\ref{eq:linear}) -- (\ref{eq:exp2}):
\begin{equation}
\label{obj:var}
\text{maximize }\quad e_2 - e_1^2.
\end{equation}

As the quadratic program (\ref{eq:linear}) -- (\ref{obj:var}) is a concave maximization problem and  can be constructed in polynomial time from $\cN$, which is of size exponential in the size of $\cM$,
we conclude that the optimal variance can be computed in exponential time.
Furthermore, from a solution to the quadratic program, we can extract a memoryless scheduler for $\cN$ as in Corollary \ref{cor:MR}.
A memoryless scheduler for $\cN$ corresponds to a {reward-based finite-memory} scheduler for $\cM$.
 This is the case here as a scheduler only has to keep track of the reward accumulated so far up to the bound $B$. Hence, we conclude:
\begin{theorem}
The maximal variance $\Var^{\max}_{\cM}(\rew)$ as well as an optimal randomized reward-based finite-memory scheduler can be computed in exponential time.
\end{theorem}

\subsection{Computation of the demonic variance}
For the demonic variance, we follow a similar line of argumentation as we did for the maximal variance. First, we restrict the class of schedulers that we have to consider to obtain the demonic variance. Afterwards, we formulate the computation of the demonic variance as a quadratic program.

To investigate the structure of schedulers necessary to obtain the demonic variance, fix a scheduler $\tsched$. We take a closer look at $\Var^{\sched,\tsched}_{\cM}(\rew)$ for a scheduler $\sched$:
\begin{align*}
& \Var^{\sched,\tsched}_{\cM}(\rew)  {=} \frac{1}{2} (\Var^{\sched}_{\cM}(\rew) {+} \Var^{\tsched}_{\cM}(\rew) {+} (\mathbb{E}^{\sched}_{\cM}(\rew) {-} \mathbb{E}^{\tsched}_{\cM}(\rew))^2) \\ 
={} &\frac{1}{2} (\mathbb{E}^{\sched}_{\cM}(\rew^2) - 2\mathbb{E}^{\sched}_{\cM}(\rew)\mathbb{E}^{\tsched}_{\cM}(\rew) + \Var^{\tsched}_{\cM}(\rew) + \mathbb{E}^{\tsched}_{\cM}(\rew)^2 )
\end{align*}
We can see that  in order to maximize $\Var^{\sched,\tsched}_{\cM}(\rew)$ for fixed $\tsched$, the scheduler $\sched$ has to maximize 
$\mathbb{E}^{\sched}_{\cM}(\rew^2) - 2\mathbb{E}^{\sched}_{\cM}(\rew)\mathbb{E}^{\tsched}_{\cM}(\rew)$.
For simplicity, abbreviate $C\eqdef \mathbb{E}^{\tsched}_{\cM}(\rew)$. 
Using the bound $M = \max_{s\in S}\mathbb{E}_{\cM,s}^{\max}(\rew)$, we know $0\leq C \leq M$ and want to find a scheduler $\sched$ maximizing 
$\mathbb{E}^{\sched}_{\cM}(\rew^2) - 2\cdot C\cdot \mathbb{E}^{\sched}_{\cM}(\rew)$.

We will  establish results analogous to the results in Section \ref{sec:structvar} for the maximization of this expression.
The proofs follow the similar ideas with slightly different calculations.
\vspace{6pt}

\paragraph*{Bound $B^\prime$}
To provide a bound $B^\prime$ in analogy to Section \ref{sec:structvar}, we use the values $Q$, $M$, and $\delta$ defined there again.
We set
\[
B^\prime = \frac{Q}{2\delta}+2M+1.
\]
As before, we know that $\Pr^{\sched}_{\cM}(\rew\geq B^\prime) \leq \frac{1}{2}$.
The following result is shown following the same idea as for Theorem \ref{thm:ermax} and is proved in Appendix \ref{app:acc}.
\begin{restatable}{theorem}{bounddem}
Let $0\leq C \leq M$.
Let $\sched$ be a scheduler for $\cM$ and let $\pi$ be a finite path with $\rew(\pi)\geq B^\prime$. Assume that $\pi$ has positive probability under $\sched$ and that  there is an action $\alpha\in \Act(s) \setminus \Act^{\max}(s)$ such that
$\sched(\pi)(\alpha)>0$. 
Let $\tsched\in \Sigma^{\max}$. Then, setting $\rsched \eqdef \sched\uparrow_{\pi,\alpha}\tsched$, we have
\[
\mathbb{E}^{\sched}_{\cM}(\rew^2) - 2\cdot C\cdot \mathbb{E}^{\sched}_{\cM}(\rew) < \mathbb{E}^{\rsched}_{\cM}(\rew^2) - 2\cdot C\cdot \mathbb{E}^{\rsched}_{\cM}(\rew) .
\]
\end{restatable}

\begin{proof}
Let $A$ be the event that scheduler $\sched$ chooses action $\alpha$ at the end of path $\pi$ and let $p$ be the probability of this event under $\sched$. Note that $\rsched$ switches to the behavior of $\tsched$ exactly in case $A$ occurs. By the definition of $B^\prime$, we know that $p<\frac{1}{2}$.
Now, we aim to estimate the difference
\begin{align*}
&\mathbb{E}^{\sched}_{\cM}(\rew^2) - 2\cdot C\cdot \mathbb{E}^{\sched}_{\cM}(\rew) - (\mathbb{E}^{\rsched}_{\cM}(\rew^2) - 2\cdot C\cdot \mathbb{E}^{\rsched}_{\cM}(\rew)) \\
={} & \mathbb{E}^{\sched}_{\cM}(\rew^2) - \mathbb{E}^{\rsched}_{\cM}(\rew^2) \\
& +2\cdot C \cdot p \cdot ( \mathbb{E}^{\tsched}_{\cM,\last(\pi)}(\rew) -  \mathbb{E}^{\sched_{\pi,\alpha}}_{\cM,\last(\pi)}(\rew))
\end{align*}

In the proof of Theorem \ref{thm:ermax}, we computed 
\begin{align*}
& \mathbb{E}^{\sched}_{\cM}(\rew^2) - \mathbb{E}^{\rsched}_{\cM}(\rew^2)  \\
 \leq {} & p ( Q - 2\cdot W \cdot  (\mathbb{E}^{\tsched}_{\cM,\last(\pi)}(\rew) -  \mathbb{E}^{\sched_{\pi,\alpha}}_{\cM,\last(\pi)}(\rew)))
\end{align*}
where  $W=\rew(\pi)$.
So, we conclude using $C\leq M$
\begin{align*}
&\mathbb{E}^{\sched}_{\cM}(\rew^2) - 2\cdot C\cdot \mathbb{E}^{\sched}_{\cM}(\rew) - (\mathbb{E}^{\rsched}_{\cM}(\rew^2) - 2\cdot C\cdot \mathbb{E}^{\rsched}_{\cM}(\rew)) \\
\leq{} & p \cdot ((2 M - 2 W ) ( \mathbb{E}^{\tsched}_{\cM,\last(\pi)}(\rew) -  \mathbb{E}^{\sched_{\pi,\alpha}}_{\cM,\last(\pi)}(\rew)) + Q )\\
\leq {} & p \cdot (2 (M-B^\prime) \delta +Q).
\end{align*}
By the choice of $B^\prime$ this value is less than $0$. 
\end{proof}

Again, we can observe in these calculations that an optimal scheduler should in fact switch to a scheduler maximizing the expected value of $\rew^2$ among the schedulers maximizing the expected value of $\rew$ above the bound $B^\prime$.

\begin{theorem}
\label{thm:optimalsched}
The demonic variance can be expressed as 
\[
\Var^{\dem}_{\cM}(\rew) = \sup_{\sched,\tsched} \Var^{\sched,\tsched}_{\cM}(\rew)
\]
where 
$\sched$ and $\tsched$ range over all schedulers  such that for all path $\pi$ with $\rew(\pi)\geq B$, we have $\sched_\pi=\usched$ for the memoryless deterministic scheduler $\usched$ given by Lemma \ref{lem:rsched}.
\end{theorem}

In order to compute the demonic variance, we use the constraints (\ref{eq:linear}) -- (\ref{eq:exp2}) with a vector of variables $x$ containing the variables $e_1$ and $e_2$ as well as a copy
(\ref{eq:linear}$^\prime$) -- (\ref{eq:exp2}$^\prime$) using variables $x^\prime$ containing $e_1^\prime$ and $e_2^\prime$.
By Lemma \ref{lem:correctquad}, for each pair of schedulers $\sched$ and $\tsched$, there is a solution to (\ref{eq:linear}) -- (\ref{eq:exp2}) and (\ref{eq:linear}$^\prime$) -- (\ref{eq:exp2}$^\prime$) 
such that
\begin{align*}
e_1 = \mathbb{E}^{\sched}_{\cM}(\rew), \quad e_2 = \mathbb{E}^{\sched}_{\cM}(\rew^2), \\
e_1^\prime = \mathbb{E}^{\tsched}_{\cM}(\rew), \quad e_2^\prime = \mathbb{E}^{\tsched}_{\cM}(\rew^2),
\end{align*}
and vice versa. By Lemma \ref{lem:variance_two}, we have
\begin{align*}
&\Var^{\sched,\tsched}_{\cM}(\rew) \\
  ={}& \frac{1}{2} \left( \Var_{\cM}^{\sched}(\rew) + \Var_{\cM}^{\tsched}(\rew) + (\mathbb{E}_{\cM}^{\sched}(\rew) - \mathbb{E}_{\cM}^{\tsched}(\rew))^2 \right) \\
={}& \mathbb{E}^{\sched}_{\cM}(\rew^2) - 2 \mathbb{E}_{\cM}^{\sched}(\rew)\mathbb{E}_{\cM}^{\tsched}(\rew) +\mathbb{E}^{\tsched}_{\cM}(\rew^2).
\end{align*}
This means we can express the demonic variance via the objective function
\begin{equation}
\text{maximize} \quad e_2 - 2e_1e_1^\prime + e_2^\prime.
\end{equation}
Together with 
constraints (\ref{eq:linear}) -- (\ref{eq:exp2}) and (\ref{eq:linear}$^\prime$) -- (\ref{eq:exp2}$^\prime$), this is a bilinear program that can be computed from $\cM$ and $\rew$ in exponential time.
Reasoning as in Corollary \ref{cor:MD_dem}, we can extract memoryless deterministic schedulers for the unfolded MDP  $\cN$, i.e., deterministic reward-based finite-memory schedulers for $\cM$.
We conlcude:

\begin{theorem}
The demonic variance $\Var^{\dem}_{\cM}(\rew)$ as well as a pair of deterministic  reward-based finite-memory schedulers $\sched$ and $\tsched$ with
$
\Var^{\dem}_{\cM}(\rew) = \Var^{\sched,\tsched}_{\cM}(\rew)$
 can be computed via a bilinear program of exponential size computable in exponential time from $\cM$.
\end{theorem}
Determining the precise complexity is left as future work.

\end{appendix}

\end{document}